\documentclass[a4paper,12pt]{article}
\usepackage{xcolor}
\usepackage{xifthen}
\usepackage{mathtools}

\def\fontsettingup{2} 

\usepackage{amsmath, amssymb}
\usepackage[T1]{fontenc}
\ifthenelse{\fontsettingup = 1}{ \usepackage{eulerpx, eucal, tgpagella}   }{}
\ifthenelse{\fontsettingup = 2}{ \usepackage{mathpazo, tgpagella} }{}

\usepackage{hyperref}
\hypersetup{
  colorlinks=true,
  linkcolor=blue,
  citecolor=blue,
  filecolor=magenta,
  urlcolor=cyan,
}

\usepackage{algorithm2e}
\usepackage{tikz}
\usetikzlibrary{arrows}
\usepackage{tcolorbox}
\usepackage{crossreftools} 

\usepackage[a4paper,margin=1in]{geometry}
\usepackage{microtype}
\usepackage{enumitem}
\usepackage{needspace}
\setlist[itemize]{leftmargin=1.5em,itemsep=0.3em,topsep=0.3em}
\allowdisplaybreaks[2]

\usepackage{amsthm, thmtools}
\usepackage[capitalise,nameinlink]{cleveref}

\newcommand{\addsharedenv}[1]{\declaretheorem[sibling=theorem]{#1}}

\theoremstyle{plain} 
\declaretheorem{theorem} 
\forcsvlist{\addsharedenv}{observation, claim, fact, lemma, proposition, corollary}

\theoremstyle{definition} 
\forcsvlist{\addsharedenv}{definition, remark, condition, assumption, example}

\ifthenelse{\fontsettingup = 1}{
  \def\*#1{\mathbf{#1}} 
  \def\+#1{\mathcal{#1}} 
  \def\-#1{\mathrm{#1}} 
  \def\^#1{\mathbb{#1}} 
  \def\!#1{\mathfrak{#1}} 
}{}

\ifthenelse{\fontsettingup = 2}{
  \def\*#1{\boldsymbol{#1}} 
  \def\+#1{\mathcal{#1}} 
  \def\-#1{\mathrm{#1}} 
  \def\=#1{\mathbb{#1}} 
  \def\!#1{\mathfrak{#1}} 
}{}

\def\oPr{\mathbf{Pr}}
\renewcommand{\Pr}[2][]{ \ifthenelse{\isempty{#1}}
  {\oPr\left[#2\right]}
  {\oPr_{#1}\left[#2\right]} } 

\def\oE{\mathbb{E}}
\newcommand{\E}[2][]{ \ifthenelse{\isempty{#1}}
  {\oE\left[#2\right]}
  {\oE_{#1}\left[#2\right]} }

\usepackage{xparse}

\def\oVar{\mathbf{Var}}
\NewDocumentCommand{\Var}{ O{} O{} m }{
  \ifthenelse{\isempty{#1}} {
    \ifthenelse{\isempty{#2}} {
      \oVar\left[#3\right]
    } {
      \oVar^{#2}\left[#3\right]
    }
  } {
    \ifthenelse{\isempty{#2}} {
      \oVar_{#1}\left[#3\right]
    } {
      \oVar_{#1}^{#2}\left[#3\right]
    }
  }
}

\def\oEnt{\mathbf{Ent}}
\NewDocumentCommand{\Ent}{ O{} O{} m }{
  \ifthenelse{\isempty{#1}} {
    \ifthenelse{\isempty{#2}} {
      \oEnt\left[#3\right]
    } {
      \oEnt^{#2}\left[#3\right]
    }
  } {
    \ifthenelse{\isempty{#2}} {
      \oEnt_{#1}\left[#3\right]
    } {
      \oEnt_{#1}^{#2}\left[#3\right]
    }
  }
}

\newcommand{\DTV}[2]{\-D_{\mathrm{TV}}\left({#1},{#2}\right)}
\newcommand{\dist}{\mathrm{dist}}
\newcommand{\e}{\mathrm{e}}
\renewcommand{\epsilon}{\varepsilon}
\renewcommand{\emptyset}{\varnothing}
\newcommand{\norm}[1]{\left\Vert#1\right\Vert}
\newcommand{\set}[1]{\left\{#1\right\}}
\newcommand{\tuple}[1]{\left(#1\right)}

\newcommand{\inner}[2]{\left\langle #1,#2\right\rangle}
\newcommand{\tp}{\tuple}

\newcommand{\abs}[1]{\left\vert#1\right\vert}

\newcommand{\Ker}{\mathsf{Ker}\;}
\newcommand{\Ran}{\mathsf{Ran}\;}
\newcommand{\indep}{\mathrel{\perp\mkern-10mu\perp}}

\title{A Spectral Local-to-Global Principle for Spin Systems on Graphs with Girth At Least Five}
\author{Xiaoyu Chen\thanks{Email: \texttt{xiaoyu@mit.edu}, 
Massachusetts Institute of Technology. Supported by the NSF CAREER grant CCF-2443045, and the Reed Fund at MIT.} \and Kuikui Liu\thanks{Email: \texttt{liukui@mit.edu}, 
Massachusetts Institute of Technology. Supported by the NSF CAREER grant CCF-2443045, and the Reed Fund at MIT.}}
\date{}

\begin{document}
\maketitle

\begin{abstract}

It is proved that, for every $\delta \in (0,1)$, the Glauber dynamics for
the uniform distribution on proper $q$-colorings is rapidly mixing when
$q \geq (1+\delta)\Delta$ and the underlying graph has girth at least $5$
and maximum degree $\Delta = \Omega_\delta(1)$. 
This result also extends to general multi-spin systems {with irreducible Glauber dynamics} satisfying a \emph{local spectral contraction} condition, including the antiferromagnetic Potts model with $q\geq (1+\delta)(1-\beta)\Delta$.
These results are achieved by a new spectral local-to-global principle on graphs with girth at least five for general multi-spin systems, and a novel Fourier analysis for Glauber dynamics on a star.
The main ideas behind all the proofs were developed through several rounds of interaction with GPT-5.6 Sol Ultra.
\end{abstract}

\clearpage
\tableofcontents
\clearpage

\section{Introduction}
Sampling from the uniform distribution over proper $q$-colorings of a graph, i.e. assignments $\chi : V \to [q]$ obeying the constraint that $\chi(u) \neq \chi(v)$ for all edges $uv \in E$, is a fundamental problem in discrete probability, statistical physics, and theoretical computer science. A long-standing folklore conjecture states that efficient sampling for arbitrary graphs of maximum degree $\Delta$ is possible as long as $q \geq \Delta + 1$ \cite{Jerrum1995}. This threshold is natural since above it, a proper $q$-coloring can be efficiently found via a local greedy procedure that sequentially assigns the vertices available colors. More importantly, it is the threshold above which physicists predict the uniform measure over proper $q$-colorings exhibits correlation decay \cite{Jonasson2002}, a phenomenon intimately related to rapid mixing of local Markov chains for sampling. It is also known that there is no efficient algorithm for approximately counting or sampling proper $q$-colorings {for even $4\leq q<\Delta$} unless $\mathsf{NP} = \mathsf{RP}$ \cite{GalanisStefankovicVigoda2015}, and so a resolution of the conjecture would imply a sharp computational phase transition.

Over the past several decades, attempts to resolve this conjecture have driven the development of a large collection of mathematical techniques in approximate counting and sampling; we provide a more comprehensive discussion in \cref{sec:prior-work}. The current state-of-the-art sampling algorithm for general graphs of maximum degree {$\Delta\geq125$} requires $q > C\Delta$ where $C \approx 1.809$ \cite{CarlsonVigoda2025}. Approaching the phase transition threshold $\Delta + 1$ has thus far required additional local sparsity assumptions on the input graph. Perhaps most directly relevant to the present work is a recent result of \cite{JMV2026} which furnishes an efficient sampler when $\Delta \geq \Omega_{\delta}(1)$, $q \geq (1 + \delta)\Delta$, and the girth of the input graph is at least $7$.

In this paper, we make further substantial progress by lowering the girth lower bound to $5$. We establish {rapid} mixing of the \emph{(heat-bath)}\footnote{The related \emph{Metropolis} Glauber dynamics also updates a single uniformly chosen vertex in each iteration. However, the chain selects a proposal color $c$ uniformly from $[q]$, and updates $\chi_{t+1}(v) \gets c$ if and only if $c$ is available to $v$ (i.e. $c \notin \{\chi_{t}(u) : u \sim v\}$); otherwise, it sets $\chi_{t+1} = \chi_{t}$. This is {the} chain analyzed in \cite{JMV2026}.} \emph{Glauber dynamics}, one of the simplest and most well-studied samplers whose evolution $\{\chi_{t}\}_{t=0}^{\infty}$ over time is described by iterating the following two-step process: Given $\chi_{t}$, we generate $\chi_{t+1}$ by
\begin{enumerate}
    \item choosing a vertex $v \in V$ uniformly, and
    \item recoloring it uniformly from the colors absent from its neighborhood. Formally, $\chi_{t+1}(u) = \chi_{t}(u)$ for all $u \neq v$, and $\chi_{t+1}(v)$ is sampled uniformly from the set $[q] \setminus \{\chi_{t}(u) : u \sim v\}$.
\end{enumerate}
This Markov chain is reversible with respect to the uniform measure over proper $q$-colorings, and ergodic when $q \geq \Delta + 2$ {\cite[Section~2]{Jerrum1995}}. Recall that its mixing time is the minimum number of steps $T$ such that for any initial $\chi_{0}$, the distribution of $\chi_{T}$ is within total variation distance $1/4$ from uniform.

\begin{theorem}[Proper Colorings]\label{thm:main}
For every $0 < \delta < 1$, there exists $\Delta_{0}(\delta) > 0$ such that the following holds: For every $n$-vertex graph with maximum degree $\Delta \geq \Delta_{0}(\delta)$ and girth at least $5$, and every integer $q \geq (1 + \delta)\Delta$, the heat-bath Glauber dynamics has mixing time {$O\left(\frac{n^2}{\delta}\log q\right)$}.
\end{theorem}
\begin{remark}[{Optimal mixing for colorings}]
{Extending our spectral-gap estimate uniformly to the list-coloring instances arising from feasible pinnings, with $n$ denoting the number of unpinned vertices, would imply spectral independence \cite{AnariEtAl2024}. Together with marginal boundedness and total connectedness for $q\geq\Delta+2$, this would yield $O_{\Delta,\delta,q}(n\log n)$ mixing \cite{ChenLiuVigoda2021}. However, the bound in \cref{thm:main} allows $\Delta$ to grow with $n$, while the $O(n\log n)$ bound is only for bounded-degree graphs ($\Delta=O(1)$).}
\end{remark}

The techniques also generalize straightforwardly to a ``soft'' version of proper coloring known as the antiferromagnetic Potts model. Fix an integer $q \geq 2$ and an edge activity parameter $0 \leq \beta \leq 1$. The Gibbs distribution $\mu = \mu_{G,q,\beta}$ of the antiferromagnetic Potts model on $G = (V,E)$ is defined as
\begin{align*}
  \mu(\sigma) \propto \beta^{\abs{\set{uv\in E : \sigma_u = \sigma_v}}}, \quad \sigma \in [q]^V.
\end{align*}
Note that when $\beta = 1$, $\mu$ becomes uniform over all of $[q]^{V}$, whereas when $\beta = 0$, $\mu$ is uniform over proper $q$-colorings.

\begin{theorem}[Antiferromagnetic Potts Model]\label{thm:Potts}
For every $0 < \delta < 1$, there exists $\Delta_{0}(\delta) > 0$ such that the following holds: For every $n$-vertex graph with maximum degree $\Delta \geq \Delta_{0}(\delta)$ and girth at least $5$, every $0 \leq \beta \leq 1$, and every integer {$q\geq2$} with $q \geq (1 + \delta)(1-\beta)\Delta$, the heat-bath Glauber dynamics for $\mu_{G,q,\beta}$ has {spectral gap $\Omega_\delta(1/n)$ and} mixing time {$O_{\Delta,\delta,q,\beta}(n^2)$}.
\end{theorem}

\subsection{General multi-spin systems}
To prove \cref{thm:Potts} (and \cref{thm:main}), we identify a general criterion of the spin interaction matrix which is sufficient to imply rapid mixing of Glauber dynamics. To state it, we define multi-spin systems in full generality. Let $q \geq 2$ be an integer. Let $A \in \mathbb{R}_{\geq0}^{q\times q}$ be a symmetric matrix, and let $\*\lambda \in \mathbb{R}_{\geq 0}^q$. The Gibbs distribution $\mu = \mu_{G,A,\*\lambda}$ with \emph{interaction matrix} $A$ and \emph{external field} $\*\lambda$ on $G$ is defined by
\begin{align*}
  \mu(\sigma) \propto \prod_{uv \in E} A_{\sigma_u,\sigma_v} \cdot \prod_{v\in V} \lambda_{\sigma_v}, \quad \forall \sigma \in [q]^V.
\end{align*}
We note that $\mu$ is invariant under a global scaling of $A$, and so without loss of generality, we will assume $A \in [0,1]^{q \times q}$ throughout.

\begin{condition}\label{cond-local-spectral-contraction}
  Let $\epsilon \in (0,1]$ and $\alpha \geq 1$.
  We say $\mu_G$ has \emph{local $(\alpha,\epsilon)$-spectral contraction} if, with $B \coloneqq \*1\*1^\intercal - A$, the following conditions hold for every $v \in V$ and every feasible pinning $\tau \in [q]^{V\setminus N[v]}$:
  \begin{itemize}
  \item it has \emph{$\alpha$-bounded marginals}:
    \begin{align*}
      \norm{B\mu^\tau_{G,v}}_{\infty} \leq \alpha/\Delta,
      \quad \text{and} \quad \norm{B\mu^\tau_{G-v,i}}_\infty \leq \alpha/\Delta, \quad \forall i \in N(v);
    \end{align*}
  \item it has \emph{$\epsilon$-contraction}:
    \begin{align*}
        \norm{
        \-{diag}\tp{\sum_{i\in N(v)} \mu^\tau_{G-v,i}}^{1/2}
        B
        \,\-{diag}\tp{\mu^\tau_{G,v}}^{1/2}
        }_{2}^2 \leq \frac{1-\epsilon}{\Delta},
    \end{align*}
  \end{itemize}
  where we treat marginal probabilities $\mu^\tau_{G,v}$ and $\mu^\tau_{G-v,i}$ as vectors in $\mathbb{R}_{\geq0}^q$. {Here $\norm{\cdot}_2$ denotes the operator norm induced by the Euclidean norm.}
  {We write $G-v=G[V\setminus\{v\}]$ and let $\mu_{G-v,i}^\tau$ denote the marginal at $i$ of the Gibbs measure on $G-v$ with the same $A,\*\lambda$, conditioned on $\tau$. See \cref{subsec:notation} for neighborhood and pinning notation.}
\end{condition}

\begin{theorem}\label{thm:general}
For every $0 < \epsilon \leq 1$ and every $\alpha \geq 1$, there exists $\Delta_{0}(\alpha,\epsilon) \geq 2\alpha$ such that the following holds: For every $n$-vertex graph $G$ with maximum degree $\Delta \geq \Delta_{0}(\alpha,\epsilon)$ and girth at least $5$, and every $A \in [0,1]^{q \times q},\*\lambda \in \mathbb{R}_{\geq0}^{q}$ such that $\mu = \mu_{G,A,\*\lambda}$ {has irreducible Glauber dynamics and} satisfies $(\alpha,\epsilon)$-spectral contraction, {the discrete-time heat-bath} Glauber dynamics with stationary measure $\mu$ has spectral gap {at least $\epsilon/(4n)$}.
\end{theorem}

\begin{remark}
The abstract framework utilized to prove \cref{thm:general} (see \Cref{thm:local-to-global,poincare-via-proj-norms}) is sufficiently general to be applicable to antiferromagnetic two-spin systems, including the hardcore and Ising models, in a fixed-gap
  $\Delta$-ary uniqueness regime.  Earlier rapid-mixing results used Weitz's
  self-avoiding-walk tree~\cite{Weitz2006} to reach the corresponding
  uniqueness regime for the $(\Delta-1)$-ary tree recursion; see, for example,
  \cite{ALO2020,ChenLiuVigodaContraction2021,ChenFengYinZhang2024,ChenChenYinZhang2025}.  For fixed positive
  uniqueness slack, the distinction between the $\Delta$-ary and
  $(\Delta-1)$-ary regimes is asymptotically negligible as $\Delta\to\infty$.
  It may also be of independent interest that our method entirely avoids the
  self-avoiding-walk tree and thereby sidesteps the obstacles to extending
  Weitz-style reductions to multi-spin systems identified
  by~\cite{LiuManiPernice2024}.
\end{remark}

\subsection{Technical overview}
In this section, 
we outline the proof of \Cref{thm:general}. 
The proof follows a direct local-to-global principle for graphs with girth at least $5$.

{\centering
\begin{tikzpicture}
    \node[align=center] (local) at (0,0) {
    \begin{tcolorbox}[width=5cm, height=2cm, halign=center, valign=center, title={{\hypersetup{linkcolor=white} \Cref{poincare-via-proj-norms}}\hfill (local)}]
    spectral gap on stars
    \end{tcolorbox}
    };
    \node[align=center] (global) at (9,0) {
    \begin{tcolorbox}[width=5.5cm, height=2cm,  halign=center, valign=center, title={\hypersetup{linkcolor=white} {\Cref{thm:general}} \hfill (global)}]
    spectral gap on graphs with girth $\geq 5$
    \end{tcolorbox}
    };

    \draw[draw=black,very thick, ->] (local) to node[midway,above, label=below:{\Cref{thm:local-to-global}}]{local-to-global} (global);
\end{tikzpicture}
\par}
\paragraph{A local-to-global principle via Bochner identity.} 
Let $\mu$ be the Gibbs distribution of a spin system on a graph $G = (V,E)$ with $n$ vertices.
Let $\+L = \sum_{v\in V} \+L_v = \sum_{v\in V} (\mathbb{I} - P_v)$ be the Laplacian matrix of the continuous-time Glauber dynamics on $\mu$.
Assuming that the continuous-time Glauber dynamics is irreducible with spectral gap $\gamma > 0$\footnote{This means the heat-bath Glauber dynamics has spectral gap $\gamma/n$.}, the Poincar{\'e} inequality is the following operator inequality in the space $(\*1^\perp, \inner{\cdot}{\cdot}_\mu)$:
\begin{align*}
    \+L &\succeq \gamma\mathbb{I} \\
    \tag{by $\+L \succ 0$ from irreducibility}
    \quad \Longleftrightarrow \quad 
    \+L^2 &\succeq \gamma \+L,
\end{align*}
where the last inequality is known as the \emph{Bochner inequality} (\Cref{prop:Bochner-identity}).

What makes the Bochner inequality different from the Poincar{\'e} inequality is that, for spin systems, it ``sees'' the underlying graph algebraically. 
For every pair of vertices $u, v$ such that $\mathrm{dist}_G(u,v) \geq 2$, by the \emph{conditional independence} property of the Gibbs distribution $\mu$, resampling $u$ and $v$ one by one conditional on the configuration on the other vertices is equivalent to resampling them together.
Algebraically, this means $P_uP_v = P_vP_u$, which further implies that $\+L_u\+L_v = \+L_v\+L_u$.
Similarly, resampling a vertex $u$ twice is equivalent to just resampling it once.
This means $P_u^2 = P_u$ and $\+L_u^2 = \+L_u$.
Hence,
\begin{align*}
\+L_u\+L_v = \+L_u\+L_u\+L_v\+L_v = (\+L_u\+L_v)^*\+L_u\+L_v \succeq 0, \quad \text{when } \mathrm{dist}_G(u,v) \geq 2.
\end{align*}

Let $\eta \in (0,1)$ be a real number.
After dropping the terms in $\+L^2$ for all pairs $u,v$ at distance at least $3$ and some fraction of the terms for pairs at distance $2$, we have
\begin{align}
    \nonumber
    \+L^2 
    &\succeq \sum_v \+L_v  + \sum_{u,v \atop uv \in E} \+L_u\+L_v  + \frac{1+\eta}{4}\sum_{u,v \atop  \mathrm{dist}_G(u,v) = 2} \+L_u\+L_v \\
    \label{eq:certificate-local-to-global}
    &={\frac12}\sum_v \tp{1-\eta\abs{N(v)}}\+L_v + \frac{1}{1+\eta}\sum_v \set{\widehat{\+L}_{v,\eta}^2 - \frac{1-\eta}{2}\widehat{\+L}_{v,\eta}}{.}
\end{align}
{Here the sums over $u,v$ are over ordered pairs,} and $\widehat{\+L}_{v,\eta} := \+L_v + \frac{1+\eta}{2}\+L_{N(v)}$ is the Laplacian of a weighted Glauber dynamics on the star $N[v]$. For brevity, we defer the detailed verification of \eqref{eq:certificate-local-to-global} to \Cref{sec:local-to-global}, where we explicitly use the fact that every pair $u,v$ at distance $2$ has exactly $1$ common neighbor as a direct consequence of the girth being at least $5$.

Then, the local-to-global principle can be directly read from \eqref{eq:certificate-local-to-global}, which says that there is a positive constant lower bound for the global spectral gap $\gamma$ as long as there exists $\eta < 1/\Delta$ such that the following local expansion can be proved for every $v$:
\begin{align} \label{eq:local}
    \widehat{\+L}_{v,\eta}^2 \succeq \frac{1-\eta}{2}\widehat{\+L}_{v,\eta},
\end{align}
where $\Delta$ is the maximum degree of the graph $G$.
Moreover, it suffices to prove \eqref{eq:local} in the spaces $(\*1^\perp, \inner{\cdot}{\cdot}_\nu)$ where $\nu = \mu^\tau_{N[v]}$ and $\tau$ ranges over all feasible pinnings in $[q]^{V\setminus N[v]}$.
This is because, by definition,
\begin{align*}
    \inner{f}{\tp{\widehat{\+L}_{v,\eta}^2 - \frac{1-\eta}{2}\widehat{\+L}_{v,\eta}} f}_\mu
    &= \E[X\sim \mu]{f(X) \cdot \left[\tp{\widehat{\+L}_{v,\eta}^2 - \frac{1-\eta}{2}\widehat{\+L}_{v,\eta}} f\right](X)} \\
    &= \oE_{\tau\sim \mu_{V\setminus N[v]}} \inner{f}{\tp{\widehat{\+L}_{v,\eta}^2 - \frac{1-\eta}{2}\widehat{\+L}_{v,\eta}} f}_{\mu^\tau},
\end{align*}
where the last equation follows from the law of total expectation.
{It suffices to prove the Poincar{\'e} inequality for the weighted Glauber dynamics with Laplacian $\widehat{\+L}_{v,\eta}$. Its positive lower bound also implies irreducibility, and hence the Bochner identity yields \eqref{eq:local}.}

\paragraph{Local expansion via a Fourier analysis.}
Fix $v \in V$ and a feasible pinning $\tau \in [q]^{V\setminus N[v]}$.
Let $\nu = \mu^\tau_{N[v]}$.
We use the idea from Fourier analysis to establish \eqref{eq:local} in $(\*1^\perp, \inner{\cdot}{\cdot}_\nu)$.
{It suffices to show that the following matrix is non-negative:}
\begin{align} \label{eq:def-Gamma}
    {\Gamma = \Gamma_\eta := \widehat{\+L}_{v,\eta} - \frac{1-\eta}{2}\mathbb{I}.}
\end{align}

For convenience, let $N(v) = \set{1,2,\cdots, d}$.
For any function $f:\mathrm{supp}(\nu) \to \mathbb{R}$, we have the following \emph{orthogonal Hoeffding expansion}:
\begin{align*}
    f = \tp{\prod_{i=1}^d \mathbb{I}} f = \prod_{i=1}^d(\+L_i + P_i) f = \sum_{\Lambda \subseteq [d]} \tp{\prod_{i\in\Lambda}\+L_i}\tp{\prod_{i\in[d]\setminus\Lambda}P_i}f =: \sum_{\Lambda \subseteq [d]} f_\Lambda,
\end{align*}
where in the third equation, we use the fact that $\mathrm{dist}_G(i,j) = 2$ so that $P_i, P_j, \+L_i, \+L_j$ commute with each other.
Analogously to the standard Fourier analysis of Boolean functions, $f_\Lambda$ depends only on the configuration at $v$ and on $\Lambda$. Intuitively, $\prod_{i\in [d]\setminus \Lambda} P_i$ resamples every coordinate $i \not\in \Lambda$ together.

Under this ``orthogonal basis'',
it is intuitive to think of $f$ as a concatenation of three vectors $f = (f_0, f_1, f_2)$, where $f_0 = f_\emptyset$, $f_1 = \sum_{i} f_{\set{i}}$, and $f_2 = f - f_0 - f_1$ represents the higher-order terms.
Then, a natural idea for proceeding with the proof is to eliminate $f_2$, which is, apparently, harder to understand.
And the \emph{Schur complement} appears to be the right tool to implement this idea (see \cref{sec:schur-complement} for a short introduction).
Under the same basis, we can also write the matrix $\Gamma$ as a $3\times 3$ block matrix.
An oversimplified high-level plan of our proof is summarized in the following display:
\begin{align*}
    \Gamma \succeq 0
    \tag{by \Cref{prop:Schur-complement-elimination}}
    &\quad { \tikz  \draw[implies-implies, double equal sign distance] (0,0) to node[midway, above] {assuming $\Gamma_{00} \succ 0$} (5,0);} \quad
    \Gamma/\Gamma_{00} \succeq 0 \\
    \tag{by \Cref{prop:Schur-complement-elimination}}
    &\quad { \tikz  \draw[implies-implies, double equal sign distance] (0,0) to node[midway, above] {assuming $(\Gamma/\Gamma_{00})_{22} \succ 0$} (5,0);} \quad
    (\Gamma/\Gamma_{00})/(\Gamma/\Gamma_{00})_{22} \succeq 0,
\end{align*}
where the last matrix is supported only on the $11$-block, which is relatively easy to verify.
{In the detailed proof, we bound $\Gamma/\Gamma_{00}$ below by a simpler matrix $N$ before eliminating the $22$-block of $N$.}
After a tedious calculation, $\Gamma = \Gamma_\eta \succeq 0$ {is ensured by}
\begin{align*}
    r := \norm{(P_v)_{00}} < 1/5 \quad \text{and} \quad \eta < 1 \quad \text{for the specific choice }  \eta = \frac{\norm{(P_v)_{10}}^2}{(1-r)(1-4r)},
\end{align*}
where $(P_v)_{00}, (P_v)_{10}$ are the $00$-block and the $10$-block of $P_v$, respectively.
We defer the detailed proof to \Cref{sec:fourier-analysis-on-star}.

Finally, we prove \Cref{thm:general} by expanding $(P_v)_{00}$ and $(P_v)_{10}$ from their definitions and bounding each of them accordingly by using the \emph{local spectral contraction} assumption. We defer the detailed proof to \Cref{sec:general}.



\subsection{Prior work}\label{sec:prior-work}

\paragraph{Rapid mixing for colorings on general graphs.}
Early coupling arguments established rapid mixing of Glauber dynamics
when $q>2\Delta$ \cite{Jerrum1995}.  Subsequent work improved the
color-to-degree ratio using flip dynamics, first to $11/6$ and then below
this threshold \cite{Vigoda2000,ChenEtAl2019}, with further improvements
through refined coupling metrics \cite{CarlsonVigoda2025}.

\paragraph{Rapid mixing for colorings on locally sparse graphs.}
Burn-in and local-uniformity arguments exploited large girth to improve
coloring bounds \cite{DyerFrieze2003,Hayes2003,Molloy2004}, eventually
yielding optimal Glauber mixing for sufficiently large constant degrees
at girth five and seven \cite{DFHV2013}.  Spectral and coupling
independence subsequently extended rapid-mixing guarantees to triangle-free
graphs and sharpened spectral-gap bounds
\cite{CGSV2021,FGYZ2022,JPV2022,ChenFeng2025}.  Non-Markovian coupling
brought the color-to-degree ratio arbitrarily close to one under stronger
girth assumptions \cite{HayesVigoda2003,JMV2026}, while tree-recursion
methods reached $q\geq\Delta+3$ for fixed $\Delta$ and sufficiently large
girth \cite{CLMM2023}.  Universality and entropy-factorization results
complement these advances by relating spectral gaps, spectral independence,
and optimal mixing \cite{AnariEtAl2024,ChenLiuVigoda2021}.

\paragraph{Rapid mixing for the antiferromagnetic Potts model.}
The zero-temperature case $\beta=0$ is the coloring model surveyed above.
At positive temperature, early work used Dobrushin-type and spatial-mixing
arguments to establish rapid Glauber mixing, with subsequent refinements
for bounded-degree graphs and lattices
\cite{SalasSokal1997,GoldbergEtAl2006Potts}.  Later work developed
correlation-decay and zero-free approaches to approximate counting
\cite{YinZhang2015SpatialPotts,BencsBerrekkalRegts2026Counting,
LiuSinclairSrivastava2025Potts}, clarified uniqueness and spatial mixing on
trees \cite{GalanisGoldbergYang2018,deBoerBuysRegts2023,BencsEtAl2023Potts,
CLMM2023,BencsBerrekkalRegts2025}, and obtained polynomial-time sampling on
random regular and sparse random graphs
\cite{BlancaEtAl2020Potts,YinZhang2016Potts,Efthymiou2022Potts}.
Complementary hardness and slow-mixing results describe the obstructions
below the natural scale $q=(1-\beta)\Delta$ and at low temperature
\cite{GalanisStefankovicVigoda2015,LiLiuYang2026Potts}.

\paragraph{Trickle-down methods.}
Beginning with Garland's method, local-to-global spectral arguments
propagate expansion from links to high-dimensional walks; this lineage now
includes local spectral expansion, spectral and entropic independence, and
scalar, matrix, partite, poset, localization, and $\mathcal C$-Lorentzian
trickle-down theorems
\cite{Garland1973,BallmannSwiatkowski1997,Zuk2003,
ParzanchevskiRosenthal2017,KaufmanMass2017,DinurKaufman2017,
Parzanchevski2017,Oppenheim2018,KaufmanOppenheim2020,Oppenheim2020,
AnariEtAl2019,CryanGuoMousa2021,AlevLau2020,ALO2020,GuoMousa2020,
FGYZ2022,JPV2022,ChenLiuVigoda2021,ALG2021,AnariEtAl2022,
KaufmanTessler2026,
AbdolazimiOveisGharan2023,GotlibKaufman2023,AnariKoehlerVuong2024, ChenEtAl2025RandomRegular,
LeakeLindbergOveisGharan2025,LeakeOveisGharan2025}.  Our local-to-global
theorem follows the same philosophy, but it uses an exact operator expansion
and the girth-five geometry to collect every interacting pair of updates
inside a closed star.

\paragraph{Bochner-type inequalities.}
Beginning with Bochner's identity and Bakry--\'Emery theory, comparisons
between a Dirichlet form and its iterated form have been used to derive
Poincar\'e and entropy-decay inequalities
\cite{Bochner1946,BakryEmery1985,BGL2014}.  Discrete versions were
developed for finite Markov chains and interacting-particle systems
\cite{BCDP2006,CaputoPosta2007,Caputo2008,CDP2009,DaiPraPosta2013,
Maas2011,ErbarMaas2012,Mielke2013,EMT2015,FathiMaas2016,EHMT2017,
ErbarFathi2018}, with more recent extensions involving Beckner and
curvature-dimension inequalities and coupling-based curvature bounds
\cite{JungelYue2017,WeberZacher2021,Pedrotti2025}.  Related
squared-generator expansions have also been applied to spectral-gap
estimates for hardcore and antiferromagnetic two-spin models
\cite{KondratievKunaOhlerich2013,GoebelEtAl2026,
GuoZhang2026Planar}.  Our Bochner identity
follows this line, using the girth-five geometry to localize pairwise
update terms on closed stars.

\subsection{Discussions about experiments with AI}
This work is mainly inspired by the recent work of G\"obel, Jenssen,
Michelen, Pappik, Perkins, and Schiller~\cite{GoebelEtAl2026}, which gives a
direct Bochner--Bakry--\'Emery proof of the spectral-gap criterion of Chen,
Chen, Chen, Yin, and Zhang~\cite{ChenEtAl2025RandomRegular}.  The original
proof of this criterion proceeds through a trickle-down theorem for field
dynamics~\cite{ChenEtAl2025RandomRegular}.

We first asked GPT-5.6 Sol Ultra to use the Bochner identity to simplify the
proofs of several results obtained by trickle-down methods, including
Oppenheim's trickle-down theorem~\cite{Oppenheim2018,AlevLau2020}, spectral independence via the pairwise
spectral-influence matrix~\cite{LeakeOveisGharan2025}, and the matrix
trickle-down theorem for sampling
edge-colorings~\cite{ALG2021}{.}\footnote{We omit the Bochner-type proofs of
these results because these results are known before and not related to the main result of this paper.}

With these proofs, we observe that the Bochner identity is useful for
reducing a global spectral-gap estimate to estimates on local gadgets:
\begin{itemize}
\item for Oppenheim's trickle-down theorem, the local gadgets are
  codimension-$2$ links;
\item for spectral independence via pairwise spectral influence, the local
  gadgets are two-part links;
\item for edge-coloring, the local gadgets are endpoint stars, which are
  cliques in the corresponding line graph.
\end{itemize}
These experiments also suggest that GPT is particularly effective at
analyzing such local gadgets.

It was then natural to ask whether this method works for proper vertex
colorings of graphs of girth at least $5$, where stars are again natural
local gadgets.
We brought this question to GPT-5.6 Sol Ultra, which proposed an affirmative
proof strategy.  The human authors then verified the argument, generalized
the proof with further assistance from GPT, and streamlined the paper.
The authors also {used} GPT-5.6 to improve the exposition and
identify typographical errors. The authors take full responsibility for the contents of this paper.

\begin{remark}[Discussion of the current obstruction]
  The existence of {$4$-cycles (not necessarily induced)} breaks the local-to-global argument in \Cref{sec:local-to-global}.
  We asked GPT-5.6 Sol Ultra to work on triangle-free graphs with the Bochner identity, but it did not find a proof after $\geq 20$ hours.

  In fact, {$4$-cycles} seem to be the only remaining hard gadgets for coloring.
  Given the present manuscript, GPT-5.6 Sol Ultra claims that it has found a proof that generalizes \Cref{thm:main} to all graphs without {$4$-cycles}.
  The proof follows a high-level plan similar to the one presented in this manuscript. For the local expansion, it uses the fact that, without {$4$-cycles}, the induced subgraph on the neighbors of any vertex is a matching.
  However, this makes the analysis of the local gadget significantly more sophisticated.
  We omit this result from this manuscript since we did not find any new ideas in it.
\end{remark}

\section{Preliminaries}
\subsection{Notation}
\label{subsec:notation}
For a nonnegative integer $m$, write $[m]\coloneqq\{1,\ldots,m\}$, with
$[0]\coloneqq\varnothing$.  Let
$G=(V,E)$ be a finite simple graph.  We use
\begin{align*}
  n&\coloneqq|V|,
  &N(v)&\coloneqq\{u\in V:\{u,v\}\in E\},
  &N[v]&\coloneqq N(v)\cup\{v\},\\
  \deg_G(v)&\coloneqq|N(v)|,
  &\Delta&\coloneqq\max_{v\in V}\deg_G(v).
\end{align*}
Thus $N(v)$ and $N[v]$ are respectively the open and closed neighborhoods
of $v$.  We write $\dist_G(u,v)$ for graph distance, and the \emph{girth}
of $G$ is the length of its shortest cycle, with girth equal to $\infty$ for a
forest.

For a probability measure $\nu$ on $\emptyset \subsetneq \Omega \subseteq [q]^{V}$, {with $\Omega=\mathrm{supp}(\nu)$ throughout,} and $S \subseteq V$, we write $\nu_{S}$ for the marginal distribution of $\sigma_{S} \coloneqq (\sigma_{v})_{v \in S}$ for $\sigma \sim \nu$. For $\Lambda \subseteq V$, an assignment $\tau \in [q]^{\Lambda}$ is a feasible pinning on $\Lambda$ if $\nu_{\Lambda}(\tau) > 0$. Putting these two pieces of notation together, we write $\nu_{S}^{\tau}$ for the law of $\sigma_{S}$ conditioned on a feasible pinning {$\tau \in [q]^\Lambda$, where $\Lambda\subseteq V\setminus S$}. If $A \subseteq \Omega$ is an event, we write $\Pr[\nu]{A} \coloneqq \Pr[\sigma \sim \nu]{\sigma \in A}$. Similarly, if $f : \Omega \to \mathbb{R}$ is a function, we write $\E[\nu]{f}$ and $\Var[\nu]{f}$ for the expectation and variance of the random variable $f(\sigma)$ for $\sigma \sim {\nu}$. We also form the inner product space
\begin{align*}
    \inner{f}{g}_\nu&\coloneqq\sum_{x\in\Omega}\nu(x)f(x)g(x), \qquad L^2(\nu)\coloneqq\bigl(\mathbb R^\Omega,\inner{\cdot}{\cdot}_\nu\bigr){.}
\end{align*}
We write $\*1_A$ for the indicator of an event $A$ and $\*1$ for the
constant-one function.  The centered subspace is
\begin{align*}
  \*1^\perp
  \coloneqq\{f\in L^2(\nu):\inner{f}{\*1}_\nu=0\}
  =\{f\in L^2(\nu):\E[\nu]{f}=0\}.
\end{align*}
{Throughout, in spectral-gap arguments (including after conditioning), we may assume without loss of generality that the support has size at least two, since the singleton case is trivial.}
We use Roman letters such as $u,v,w,i,j$ for vertices and Fraktur letters
such as $\mathfrak a,\mathfrak b,\mathfrak c$ for colors.

Subscripts in $O_\delta(\cdot)$, $\Omega_\delta(\cdot)$, and
$o_\delta(\cdot)$ allow the implicit constants, thresholds, or rates of
convergence to depend on $\delta$.  In particular,
$\Delta=\Omega_\delta(1)$ means that $\Delta$ is sufficiently large in
terms of $\delta$.

For a finite set $S$, $\binom{S}{2}$ denotes the collection of its two-element
subsets.

\subsection{General inner product space}
Throughout this subsection, all vector spaces are finite-dimensional.  An
inner-product space is a real vector space $H$ equipped with a symmetric
bilinear form $\inner{x}{y}_H$ such that $\inner{x}{x}_H>0$ for every
$x\neq0$.  It induces the norm $\norm{x}_H\coloneqq\sqrt{\inner{x}{x}_H}$.
For two vectors $f,g\in H$, we say $f\perp g$ if $\inner{f}{g}_H=0$.
For subspaces $K_0,K_1\subseteq H$, the notation $K_0\perp K_1$ means
that every vector in $K_0$ is orthogonal to every vector in $K_1$.  The
orthogonal complement of $K\subseteq H$ is
$K^\perp\coloneqq\{x\in H:x\perp y\text{ for every }y\in K\}$.
\begin{remark}
  For the rest of this {paper}, once the inner-product space is clear from
  context, we use $\inner{\cdot}{\cdot}$ for its inner product and
  $\norm{\cdot}$ for the corresponding norm.
\end{remark}

For a subspace $K\subseteq H$, write $\operatorname{proj}_K$ for the
orthogonal projection onto $K$.  For a linear operator $A:H\to H$, write
$A|_K$ for its restriction to $K$, and write
\begin{align*}
  \Ker(A)&\coloneqq\{x\in H:Ax=0\},
  &\Ran(A)&\coloneqq\{Ax:x\in H\}.
\end{align*}
We write $\mathrm{Id}$ or $\mathbb I$ for the identity operator when its
space is clear.  A subscript indicates the space on which the identity
acts; for example, $\mathbb I_i$ is the identity on a subspace indexed by
$i$.

\begin{definition}
  Let $A, A^*:H\to H$ be linear operators.
  We say $A^*$ is an adjoint operator for $A$ if, for every $f,g\in H$,
  \begin{align*}
    \inner{Af}{g} = \inner{f}{A^*g}.
  \end{align*}
\end{definition}

\begin{proposition}
  Let $A:H\to H$ be a linear operator.
  Then $A^*$ exists and is unique.
\end{proposition}
The proof is standard and is omitted.

We will also use the operator analogue of positive semidefinite matrices.
\begin{definition}
  For a linear operator $A:H\to H$:
  \begin{itemize}
  \item $A$ is \emph{self-adjoint} if $A=A^*$;
  \item $A\succeq0$ means $A=A^*$ and
    $\inner{x}{Ax}_H\geq0$ for every $x\in H$;
  \item $A\succ0$ means $A=A^*$ and
    $\inner{x}{Ax}_H>0$ for every $x\neq0$;
  \item $A\succeq B$ means $A-B\succeq0$;
  \item $A\preceq B$ means $B\succeq A$;
  \item $\norm{A}\coloneqq{\max_{\norm{x}_H\leq1}}\norm{Ax}_H$ is the operator norm.
  \end{itemize}
  Thus $\succeq$ is an order on self-adjoint operators, not an entry-wise
  inequality.  All spectra below consist of eigenvalues; multiplicities are
  stated explicitly when relevant.
\end{definition}

\begin{fact}
  For every linear operator $A:H\to H$, we have
  $\norm{AA^*} = \norm{A^*A} = \norm{A}^2$.
\end{fact}

\subsection{Schur complement} \label{sec:schur-complement}
Let $H=H_0\oplus H_1$ be an orthogonal decomposition: every $x\in H$ has a
unique representation $x=x_0+x_1$, with $x_i\in H_i$ and $H_0\perp H_1$.
For a linear operator $A:H\to H$, its blocks are
\[
  A_{ij}\coloneqq\operatorname{proj}_{H_i}A|_{H_j}:H_j\longrightarrow H_i.
\]
If $A=A^*$, then
\[
  A=
  \begin{pmatrix}
    A_{00}&A_{01}\\
    A_{10}&A_{11}
  \end{pmatrix},
  \qquad
  A_{00}=A_{00}^*,\quad A_{11}=A_{11}^*,\quad A_{10}=A_{01}^*.
\]

\begin{definition}
Assume $A_{00}$ is invertible.
The \emph{Schur complement} of $A_{00}$ in $A$ is
\begin{align} \label{eq:abstract-Schur-complement}
  A/A_{00}
  \coloneqq A_{11}-A_{10}A_{00}^{-1}A_{01}.
\end{align}
\end{definition}

The following result is the main motivation behind the Schur complement.
\begin{proposition} \label{prop:Schur-complement-elimination}
  If $A=A^*$ and $A_{00}\succ0$, then
  \begin{align} \label{eq:Schur-positivity-equivalence}
    A\succeq0
    \quad\Longleftrightarrow\quad
    A/A_{00}\succeq0.
  \end{align}
\end{proposition}
\begin{proof}
Note that $A/A_{00}: H_1 \to H_1$ by definition.
Fix $z\in H_1$ and minimize
\[
  \langle h+z,A(h+z)\rangle
  \qquad\text{over }h\in H_0.
\]
Since $A_{00}\succ0$, the unique minimizer is characterized by
\[
  \operatorname{proj}_{H_0}A(h+z)
  =A_{00}h+A_{01}z=0,
\]
and hence $h=-A_{00}^{-1}A_{01}z$. Substitution gives the
variational identity
\[
  \min_{h\in H_0}\langle h+z,A(h+z)\rangle
  =\langle z,(A/A_{00})z\rangle .
\]
Because every vector in $H$ has the form $h+z$, it follows that
\begin{align*}
  A\succeq0
  &\quad\Longleftrightarrow\quad
  \min_{h\in H_0}\langle h+z,A(h+z)\rangle\geq0
  \ \text{for every }z\in H_1 \\
  &\quad\Longleftrightarrow\quad
  A/A_{00}\succeq0 . \qedhere
\end{align*}
\end{proof}

\begin{remark}
The Schur complement $A/A_{00}$ is the effective operator on $H_1$
obtained after choosing the $H_0$-component optimally.
\Cref{prop:Schur-complement-elimination}
therefore reduces positivity of an operator on $H_0\oplus H_1$ to
positivity of an operator on $H_1$. In our applications, $H_1$ has
additional structure that makes the reduced inequality easier to analyze,
even when the explicit expression for $A/A_{00}$ is more complicated.
\end{remark}

\subsection{Spectral gap, mixing time, and a Bochner identity}
Let $\Omega$ be a finite state space, let $\mu$ be a probability
distribution with full support on $\Omega$, and let $P$ be the transition
operator of a Markov chain that is reversible with respect to $\mu$.
Thus $P$ is self-adjoint on $L^2(\mu)$ and $P\*1=\*1$.  We assume that
$P\succeq0$, as will be the case for the Glauber dynamics below, and
define its discrete-time Laplacian by
\begin{align*}
  \+L\coloneqq\mathrm{Id}-P.
\end{align*}

Let
\begin{align*}
  1=\lambda_1(P)\geq\lambda_2(P)\geq\cdots\geq0
\end{align*}
be the eigenvalues of $P$, counted with multiplicity.  The
\emph{spectral gap} of $P$ is $1-\lambda_2(P)$.  It is positive exactly
when the chain is irreducible.  The ratio
\begin{align*}
  \frac{\inner{f}{\+L f}_\mu}{\inner{f}{f}_\mu}
\end{align*}
is the Rayleigh quotient of $\+L$ at $f$.  By the Rayleigh--Ritz
variational principle,
\begin{align*}
  1-\lambda_2(P)
  =\inf_{\substack{f\neq0\\f\perp\*1}}
    \frac{\inner{f}{\+L f}_\mu}{\inner{f}{f}_\mu}.
\end{align*}
Equivalently, $1-\lambda_2(P)$ is the largest constant $\gamma$ for
which the \emph{Poincar{\'e} inequality}
\begin{align*}
  \gamma\inner{f}{f}_\mu\leq\inner{f}{\+L f}_\mu,
  \qquad f\perp\*1,
\end{align*}
holds.
The main tool of this {paper} is the following Bochner-type identity.
\begin{proposition} \label{prop:Bochner-identity}
  If $P$ is irreducible, then we have
  \begin{align*}
    \gamma
    \coloneqq\inf_{\substack{f\neq0\\f\perp\*1}}
      \frac{\inner{f}{\+L f}_\mu}{\inner{f}{f}_\mu}
    =\inf_{\substack{f\neq0\\f\perp\*1}}
      \frac{\inner{\+L f}{\+L f}_\mu}{\inner{f}{\+L f}_\mu}.
  \end{align*}
\end{proposition}
\begin{proof}
  We restrict to the inner-product space
  $(\*1^\perp,\inner{\cdot}{\cdot}_\mu)$.
  The left-hand side means
  $\+L\succeq\gamma\mathrm{Id}$, whereas the right-hand side means
  $\+L^2\succeq\gamma\+L$.  Since $P$ is irreducible,
  $\+L\succ0$ on $\*1^\perp$, and hence $\+L$ is invertible there.  On
  this subspace, $\+L^{1/2}$ denotes the unique positive-definite square
  root of $\+L$, and $\+L^{-1/2}$ denotes its inverse.  Conjugating
  $\+L^2\succeq\gamma\+L$ by $\+L^{-1/2}$ gives
  $\+L\succeq\gamma\mathrm{Id}$, and the converse follows by conjugating
  with $\+L^{1/2}$.
\end{proof}

For an irreducible chain, its mixing time is
\begin{align*}
  t_{\mathrm{mix}}(\varepsilon)
  \coloneqq\min\left\{t\in\mathbb Z_{\geq0}:
    \max_{x\in\Omega}\DTV{P^t(x,\cdot)}{\mu}
    \leq\varepsilon\right\},
\end{align*}
where for two probability
distributions $\nu_1$ and $\nu_2$ on $\Omega$, their total variation
distance is
\[
  \DTV{\nu_1}{\nu_2}
  \coloneqq\frac12\sum_{\sigma\in\Omega}
    |\nu_1(\sigma)-\nu_2(\sigma)|.
\]
If its spectral gap is $\gamma>0$, then, for
$0<\varepsilon\leq1/2$, positive semidefiniteness of $P$ and $L^2$
contraction give the standard worst-start bound
\cite[Theorem~12.4]{LevinPeres2017}
\begin{align}\label{eq:spgap-implies-mixing-time}
  t_{\mathrm{mix}}(\varepsilon)
  \leq\left\lceil\frac1\gamma\left(
    \log\frac1{2\varepsilon}
    +\frac12\log\frac1{\min_{x\in\Omega}\mu(x)}
  \right)\right\rceil.
\end{align}

\subsection{Glauber dynamics and its Laplacians}
We now specialize the preceding discussion to heat-bath Glauber
dynamics.
Let $V$ be a finite coordinate set with $n=|V|$, let $\mu$ be a
distribution whose support is $\Omega\subseteq[q]^V$, and let
$X\sim\mu$.  For each $v\in V$, define the single-site heat-bath operator
$P_v$ point-wise by
\begin{align} \label{def-Pi}
  (P_vf)(x)
  \coloneqq\E[\mu]{f(X)\mid
    X_{V\setminus\{v\}}=x_{V\setminus\{v\}}},
  \qquad x\in\Omega.
\end{align}
Since the state space is finite, $P_v$ can be viewed as a matrix in
$\mathbb{R}^{\Omega\times\Omega}$.  Define the single-site Laplacian by
\begin{align} \label{def-Li}
  \+L_v\coloneqq\mathrm{Id}-P_v.
\end{align}
For any $B\subseteq V$, write
\begin{align} \label{def-LB}
  \+L_B\coloneqq\sum_{v\in B}\+L_v.
\end{align}
All these operators act on the inner-product space $L^2(\mu)$ defined in
\Cref{subsec:notation}.

\begin{fact}
  The operators $P_v$ and $\+L_v$ are self-adjoint orthogonal
  projections in $L^2(\mu)$.
\end{fact}

The transition operator of the discrete-time Glauber dynamics is
\begin{align*}
  P_{\mathrm{GD}}\coloneqq\frac1n\sum_{v\in V}P_v.
\end{align*}
Define
\begin{align} \label{def-L}
  \+L\coloneqq\sum_{v\in V}\+L_v.
\end{align}
Then
\begin{align*}
  P_{\mathrm{GD}}=\mathrm{Id}-\frac1n\+L,
\end{align*}
so the Laplacian of the discrete-time Glauber dynamics is $\+L/n$.
Since $P_{\mathrm{GD}}$ is an average of orthogonal projections, it is
self-adjoint and positive semidefinite, and the preceding discussion
applies.  In particular, the spectral gap of $P_{\mathrm{GD}}$ is
$1/n$ times the smallest eigenvalue of $\+L$ on $\*1^\perp$.

In the continuous-time Glauber dynamics, every coordinate is updated at
rate one.  Its generator is $-\+L$, so $\+L$ is its Laplacian and its
spectral gap is $n$ times that of the discrete-time Glauber dynamics.

Given nonnegative update rates $a_v$, the Laplacian of the corresponding
weighted continuous-time dynamics is $\sum_{v\in V}a_v\+L_v$.  We call
this dynamics irreducible when the kernel of its Laplacian consists only
of constant functions.  The proof of \Cref{prop:Bochner-identity}
applies verbatim to every such irreducible weighted dynamics.
\section{A local-to-global principle for girth-5 spin systems} \label{sec:local-to-global}
Let $G=(V,E)$ be a graph with girth at least $5$, and use the Gibbs and
pinning conventions from \Cref{subsec:notation}.

\begin{definition}
Given a distribution on configurations on the vertex set $N[v]$, whose
induced subgraph is a star, and a parameter $0\leq\eta_v<1$, the
\emph{weighted continuous-time Glauber
dynamics with parameter $\eta_v$} has Laplacian
\begin{align*}
  \+L_v+\frac{1+\eta_v}{2}\+L_{N(v)}.
\end{align*}
Here the single-site Laplacians are formed with respect to the given
distribution: the center is updated at rate one, and each leaf is updated
at rate $(1+\eta_v)/2$.
\end{definition}

Throughout this paper, we assume that every Gibbs distribution satisfies the following conditional independence property.
\begin{definition}\label{def:cond-indep}
  A Gibbs distribution $\mu$ on $G$ is \emph{conditionally independent} if
  \begin{align*}
    X_u\indep X_v\mid X_{V\setminus\{u,v\}}
    \qquad\text{whenever }\dist_G(u,v)\geq2,
  \end{align*}
where we use $\indep$ to denote independence.
\end{definition}

\begin{theorem} \label{thm:local-to-global}
Let $G=(V,E)$ be a graph with girth at least $5$, and let $\mu$ be a conditionally independent Gibbs
distribution on $G$ whose Glauber dynamics is irreducible.  Suppose that
the parameters $0\leq\eta_v<1$ satisfy the following conditions:
\begin{itemize}
\item for every $v\in V$ and every feasible pinning $\tau$ outside $N[v]$, the
weighted continuous-time Glauber dynamics for $\mu_{N[v]}^\tau$ with
parameter $\eta_v$ has spectral gap at least $(1-\eta_v)/2$;
\item there exists $\epsilon>0$ such that
$\sum_{u\in N(v)}\eta_u\leq1-\epsilon$ for every $v\in V$.
\end{itemize}
Then the continuous-time Glauber dynamics for $\mu$ has spectral gap at
least $\epsilon/2$.
\end{theorem}

\begin{proof}
Fix $f\in L^2(\mu)$.  Expanding the square, with all pairs below
unordered, gives
\begin{align*}
  \norm{\+L f}^2
  &=\sum_{v\in V}\norm{\+L_vf}^2
    +2\sum_{\{u,w\}\in\binom{V}{2}}
      \inner{\+L_u f}{\+L_w f}.
\end{align*}
If $\dist_G(u,w)\geq2$, conditional independence implies that the two
update projections commute.  Since $\+L_u$ and $\+L_w$ are commuting
orthogonal projections,
\begin{align*}
  \inner{\+L_u f}{\+L_w f}
  =\norm{\+L_u\+L_wf}^2\geq0.
\end{align*}
We therefore discard the terms with distance at least three from $\norm{\+Lf}^2$.
In other words,
\begin{align*}
  \norm{\+L f}^2
  &\geq\sum_{v\in V}\norm{\+L_vf}^2
    +2\sum_{\{u,w\}\in\binom{V}{2} \atop 1 \leq \mathrm{dist}_G(u,w) \leq 2}
    \inner{\+L_u f}{\+L_w f}.
\end{align*}

Because $G$ has no triangles, the vertices in $N(v)$ are pairwise
nonadjacent.  Thus
\begin{align*}
  \+L_{N(v)}^2-\+L_{N(v)}
  =2\sum_{\set{u,w}\in\binom{N(v)}{2}}\+L_u\+L_w\succeq0.
\end{align*}
Because $G$ has no $4$-cycles, every pair at distance two has a unique
common neighbor.
After dropping vertex pairs with distance at least three, a direct double-counting argument now gives
\begin{align*}
  \norm{\+Lf}^2
  \geq\sum_{v\in V}\underbrace{\left(
    \inner{f}{\+L_v f}
    +\inner{\+L_vf}{\+L_{N(v)}f}
    +\inner{f}{(\+L_{N(v)}^2-\+L_{N(v)})f}
  \right)}_{=:T_v}.
\end{align*}

A direct expansion shows that
\begin{align*}
  &T_v-\frac1{1+\eta_v}
  \inner{f}{\left(\+L_v+\frac{1+\eta_v}{2}\+L_{N(v)}\right)^2f}\\
  &\qquad=\inner{f}{\left(
    \frac{\eta_v}{1+\eta_v}\+L_v
    +\frac{3-\eta_v}{4}\+L_{N(v)}^2
    -\+L_{N(v)}
  \right)f}.
\end{align*}
Condition on $X_{V\setminus N[v]}$. {The positive local gap implies irreducibility, and subtracting the conditional mean of $f$ leaves both Laplacian quadratic forms unchanged since the local Laplacian annihilates constants.} For each feasible pinning $\tau$, apply
the weighted spectral-gap assumption and \Cref{prop:Bochner-identity} in
$L^2(\mu_{N[v]}^\tau)$, and then average over $\tau$.  Conditioning on each such pinning,
the restriction of every global single-site Laplacian with index in $N[v]$
is the corresponding single-site Laplacian for $\mu_{N[v]}^\tau$.  The
same proof applies because the local Laplacian is a nonnegative weighted
sum of single-site Laplacians.  This gives
\begin{align*}
  &\frac1{1+\eta_v}
  \inner{f}{\left(\+L_v+\frac{1+\eta_v}{2}\+L_{N(v)}\right)^2f}\\
  &\qquad\geq\frac{1-\eta_v}{2(1+\eta_v)}
  \inner{f}{\left(\+L_v+\frac{1+\eta_v}{2}\+L_{N(v)}\right)f}.
\end{align*}
Combining the last two displays, and using
$\+L_{N(v)}^2-\+L_{N(v)}\succeq0$, yields
\begin{align*}
  T_v
  &\geq\inner{f}{\left(
    \frac12\+L_v-\frac{\eta_v}{2}\+L_{N(v)}
    +\frac{3-\eta_v}{4}
      (\+L_{N(v)}^2-\+L_{N(v)})
  \right)f}\\
  &\geq\inner{f}{\left(
    \frac12\+L_v-\frac{\eta_v}{2}\+L_{N(v)}
  \right)f}.
\end{align*}
Summing over $v$ and collecting the coefficient of each single-vertex
energy gives
\begin{align*}
  \norm{\+Lf}^2
  &\geq\frac12\sum_{v\in V}
    \left(1-\sum_{u\in N(v)}\eta_u\right)
    \inner{f}{\+L_vf}\\
  &\geq\frac\epsilon2\inner{f}{\+Lf}.
\end{align*}
The Bochner identity and irreducibility finish the proof.
\end{proof}

\section{Fourier analysis for Glauber dynamics on a star} \label{sec:fourier-analysis-on-star}
Let $G = (V, E)$ be a star.
For convenience, let $v$ be the center of the star and let
$N(v)=\set{1,\dots,d}$.
Let $\nu$ be a distribution on $[q]^V$, and let $X\sim\nu$.  Assume that,
conditional on $X_v$, the leaves $(X_i)_{i\in N(v)}$ are independent.
Let $P_u$ and $\+L_u$ denote the heat-bath projection
and single-site Laplacian in $L^2(\nu)$ for every $u\in V$, as defined in
\Cref{def-Pi,def-Li}.
In this section, we will work 
with the inner-product space
\begin{align*}
  (\*1^\perp, \inner{\cdot}{\cdot}_\nu).
\end{align*}

We give a general method to bound the spectral gap of
weighted continuous-time Glauber dynamics with parameter $\eta$.
Conditional independence of the leaves implies that their heat-bath
projections and single-site Laplacians commute.  Hence the order of the
products in the following definition is immaterial.

\begin{definition} \label{def:Hoeffding-expand}
Let $f \in \*1^\perp$ be a function.
Its \emph{orthogonal Hoeffding expansion} is
\begin{align*}
  f = \sum_{\Lambda \subseteq [d]} f_\Lambda, \quad \text{where} \quad f_\Lambda \coloneqq \tp{\prod_{i\in\Lambda}\+L_i}\tp{\prod_{i\in[d]\setminus\Lambda}P_i}f.
\end{align*}
Here an empty product of operators is the identity.
We call $|\Lambda|$ the \emph{degree} of the component $f_\Lambda$.
\end{definition}

\begin{remark}
    The decomposition in \Cref{def:Hoeffding-expand} is well-defined since
    \begin{align*}
        \mathbb{I} = \prod_i \mathbb{I} = \prod_i\tp{\+L_i + P_i} = \sum_{\Lambda \subseteq [d]} \tp{\prod_{i\in\Lambda}\+L_i}\tp{\prod_{i\in[d]\setminus\Lambda}P_i},
    \end{align*}
    where in the last equation, we use the fact that, for all $i\neq j$, conditional independence of the leaves given the center implies that
  $\+L_i,P_i,\+L_j,P_j$ all commute.
\end{remark}

We collect some properties of the orthogonal Hoeffding expansion.
\begin{proposition} \label{prop:Hoeffding-expand}
  The following hold:
  \begin{itemize}
  \item for all distinct $A,B\subseteq[d]$, we have $f_A\perp f_B$;
  \item $\+L_i f_\Lambda = \*1_{\set{i\in\Lambda}}f_\Lambda$ and
    $\+L_{N(v)} f_\Lambda = \abs{\Lambda} f_\Lambda$.
  \end{itemize}
\end{proposition}
\begin{proof}
  For all $i\in[d]$, we have $\+L_iP_i=P_i-P_i^2=0$, since $P_i$ is
  an orthogonal projection.  Also, for all $i\neq j$, conditional
  independence of the leaves given the center implies that
  $\+L_i,P_i,\+L_j,P_j$ all commute.
  The two bullets follow directly.
\end{proof}

We partition $\*1^\perp$ into three orthogonal subspaces:
\begin{align*}
  \*1^\perp = K_0 \oplus K_1 \oplus K_2,
\end{align*}
where $K_0$ contains the degree-zero component $f_\emptyset$, $K_1$
contains the degree-one components $\sum_i f_{\set{i}}$, and $K_2$
contains the remaining higher-order components.
By \Cref{prop:Hoeffding-expand}, we can formally define them as
\begin{align*}
  K_0 &\coloneqq \Ker (\+L_{N(v)}), \\
  K_1 &\coloneqq \Ker (\+L_{N(v)} - \mathbb{I}), \\
  K_2 &\coloneqq \bigoplus_{k=2}^d
    \Ker (\+L_{N(v)}-k\mathbb{I}).
\end{align*}
When $d<2$, the direct sum defining $K_2$ is the zero subspace.
For convenience, let $K_+ \coloneqq K_1 \oplus K_2$.
For $i\in\set{0,1,2,+}$, let $\mathbb I_i$ be the identity operator on
$K_i$.
For the rest of this section, we regard every operator as a block operator
with block indices $\set{0,1,2}$.
An index $+$ refers to the combined subspace $K_+=K_1\oplus K_2$.

\begin{definition} \label{block-matrix-via-Ki}
  For a linear map $A:\*1^\perp\to\*1^\perp$ and
  $i,j\in\set{0,1,2,+}$, define its projected block by
  \begin{align*}
    A_{ij} \coloneqq \operatorname{proj}_{K_i} A \vert_{K_j}
    : K_j \longrightarrow K_i.
  \end{align*}
\end{definition}

\begin{theorem} \label{poincare-via-proj-norms}
  Suppose that $r \coloneqq \norm{(P_v)_{00}} < 1/5$ and
  \[\eta \coloneqq \frac{\norm{(P_v)_{10}}^2}{(1-r)(1-4r)} < 1.\]
  Then the following Poincar{\'e} inequality holds with $\eta${:}
  \begin{align} \label{target:spgap-weighted-GD}
    \frac{1-\eta}{2}\inner{f}{f} \leq \inner{f}{\tp{\+L_v + \frac{1+\eta}{2}\+L_{N(v)}} f}, \quad f \in \*1^\perp.
  \end{align}
\end{theorem}

The proof of \Cref{poincare-via-proj-norms} is given in
\Cref{subsec:star-schur-complement-proof}.

\subsection{A Schur-complement proof}
\label{subsec:star-schur-complement-proof}
We now prove \Cref{poincare-via-proj-norms}.
Recall that we work in the inner-product space $(\*1^\perp, \inner{\cdot}{\cdot}_\nu)$.
Then \eqref{target:spgap-weighted-GD} can be rewritten as a matrix inequality.
It suffices to show the following
operator inequality on $(\*1^\perp,\inner{\cdot}{\cdot})$:
\begin{align*}
  {\Gamma} \coloneqq \+L_v + \frac{1+\eta}{2}\+L_{N(v)} - \frac{1-\eta}{2} \mathbb{I} \succeq 0.
\end{align*}
By \Cref{block-matrix-via-Ki},
${\Gamma}$ can be written as
\begin{align*}
  {\Gamma} &= \begin{bmatrix}
    (\+L_v)_{00} - \frac{1-\eta}{2}\mathbb{I}_0 & (\+L_v)_{0+} \\
    (\+L_v)_{+0} & (\+L_v)_{++} - \frac{1-\eta}{2}\mathbb{I}_+ + \frac{1+\eta}{2} (\+L_{N(v)})_{++}
  \end{bmatrix}
\end{align*}
where $\+L_{N(v)}$ vanishes on $K_0$ and preserves $K_+$ by definition.

\begin{proof}[Proof of \Cref{poincare-via-proj-norms}]
Recall that $0\leq\eta<1$.
Since $(\+L_v)_{00}=\mathbb{I}_0-(P_v)_{00}$, we have
\begin{align*}
  {\Gamma}_{00}=(\+L_v)_{00}-\frac{1-\eta}{2}\mathbb{I}_0
  =\frac{1+\eta}{2}\mathbb{I}_0-(P_v)_{00}\succ0.
\end{align*}
Thus \Cref{prop:Schur-complement-elimination} gives
\begin{align*}
  {\Gamma} \succeq 0 \quad \Longleftrightarrow \quad {\Gamma}/{\Gamma}_{00} \succeq 0.
\end{align*}
Hence it remains to prove ${\Gamma}/{\Gamma}_{00}\succeq0$.  A direct calculation gives
\begin{align*}
  {\Gamma}/{\Gamma}_{00}
  &= {\Gamma}_{++} - {\Gamma}_{+0}{\Gamma}_{00}^{-1}{\Gamma}_{0+} \\
  &= (\+L_v)_{++} - \frac{1-\eta}{2}\mathbb{I}_+
     + \frac{1+\eta}{2} (\+L_{N(v)})_{++} \\
  &\quad - (\+L_v)_{+0}
    \tp{(\+L_v)_{00} - \frac{1-\eta}{2}\mathbb{I}_0}^{-1}
    (\+L_v)_{0+} \\
  &= \+L_v/(\+L_v)_{00} - \frac{1-\eta}{2}\mathbb{I}_+ + \frac{1+\eta}{2} (\+L_{N(v)})_{++} \\
  &\quad + (\+L_v)_{+0}\left((\+L_v)_{00}^{-1}
    - \tp{(\+L_v)_{00} - \frac{1-\eta}{2}\mathbb{I}_0}^{-1}\right)
    (\+L_v)_{0+} \\
  &= \+L_v/(\+L_v)_{00} - \frac{1-\eta}{2}\mathbb{I}_+ + \frac{1+\eta}{2} (\+L_{N(v)})_{++} \\
  &\quad - \frac{1-\eta}{2} (\+L_v)_{+0}(\+L_v)_{00}^{-1}
    \tp{(\+L_v)_{00} - \frac{1-\eta}{2}\mathbb{I}_0}^{-1}
    (\+L_v)_{0+}.
\end{align*}
Since $\+L_v\succeq0$ and $(\+L_v)_{00}\succ0$, another application of
\Cref{prop:Schur-complement-elimination} gives
$\+L_v/(\+L_v)_{00}\succeq0$.  Moreover,
\begin{align*}
  (\+L_v)_{+0}(\+L_v)_{00}^{-1}
  \tp{(\+L_v)_{00} - \frac{1-\eta}{2}\mathbb{I}_0}^{-1}
  (\+L_v)_{0+}
  &\preceq
  \frac{(P_v)_{+0}(P_v)_{0+}}
  {(1-r)(\frac{1+\eta}{2}-r)}.
\end{align*}
Here $(\+L_v)_{0+}=-(P_v)_{0+}$ and
$(\+L_v)_{+0}=-(P_v)_{+0}$, so the signs cancel.  Also,
\Cref{prop:Hoeffding-expand} gives
\begin{align*}
  (\+L_{N(v)})_{++} \succeq
  \begin{bmatrix}
    \mathbb{I}_1 & 0 \\
    0 & 2 \mathbb{I}_2
  \end{bmatrix}.
\end{align*}
Combining these bounds and dropping $\eta$ from the feedback coefficient
gives
\begin{align*}
  {\Gamma}/{\Gamma}_{00} \succeq
  \begin{bmatrix}
    \eta \mathbb{I}_1 & 0 \\
    0 & \frac{1}{2} \mathbb{I}_2
  \end{bmatrix} - \underbrace{\frac{1}{(1-r)(1-2r)}
  \begin{bmatrix}
    (P_v)_{10} \\ (P_v)_{20} 
  \end{bmatrix}
  \begin{bmatrix}
    (P_v)_{01} & (P_v)_{02} 
  \end{bmatrix} }_{=:F}.
\end{align*}
Denote the right-hand side by $N$.  We regard $F$ and $N$ as block
operators on $K_+=K_1\oplus K_2$, and
write $F_{ij}$ and $N_{ij}$ for their blocks, where $i,j\in\{1,2\}$.
We first claim that
\begin{align} \label{eq:bound-||F||}
  0 \preceq F \preceq \frac{r}{1-2r}\mathbb{I}_+.
\end{align}
Since $r<1/5$, \eqref{eq:bound-||F||} implies
$N_{22}\succeq(1/2-r/(1-2r))\mathbb{I}_2\succ0$.
  After applying the Schur complement with $K_2$ as the eliminated subspace, it
remains to show $N/N_{22}\succeq0$, where
$N/N_{22}\coloneqq N_{11}-N_{12}N_{22}^{-1}N_{21}$.  We have
\begin{align*}
  N/N_{22}
  &= \eta \mathbb{I}_1 - F_{11} - F_{12}N_{22}^{-1}F_{21} \\
  &\succeq \eta\mathbb{I}_1 - \norm{F_{11}}\mathbb{I}_1 - \frac{\norm{F_{12}}^2}{\frac{1}{2} - \norm{F}}\mathbb{I}_1.
\end{align*}
Since $F\succeq 0$, we have $\norm{F_{12}}^2 \leq \norm{F}\norm{F_{11}}$. This gives
\begin{align*}
  N/N_{22}
  &\succeq \eta\mathbb{I}_1 - \norm{F_{11}}\mathbb{I}_1 - \frac{\norm{F_{11}}\norm{F}}{\frac{1}{2} - \norm{F}}\mathbb{I}_1 \\
  &= \eta\mathbb{I}_1
  - \frac{\norm{(P_v)_{10}}^2}{(1-r)(1-2r)}
    \frac{1}{1-2\norm{F}}\mathbb{I}_1 \\ 
  &\succeq \left(\eta
    - \frac{\norm{(P_v)_{10}}^2}{(1-r)(1-4r)}\right)
    \mathbb{I}_1
  \succeq0.
\end{align*}
The last two steps use \eqref{eq:bound-||F||} and the choice of $\eta$.

It remains to prove \eqref{eq:bound-||F||}.  By definition,
\begin{align*}
  \norm{F} \leq
  \frac{\norm{(P_v)_{+0}}^2}{(1-r)(1-2r)}.
\end{align*}
Since $P_v^2=P_v$,
\begin{align*}
  (P_v)_{00}=(P_v)_{00}^2+(P_v)_{0+}(P_v)_{+0}.
\end{align*}
Since $0\preceq(P_v)_{00}\preceq\mathbb I_0$ and
$\norm{(P_v)_{00}}=r<1/2$, the spectrum of $(P_v)_{00}$ lies in $[0,r]$.
Therefore
\begin{align*}
  \norm{(P_v)_{+0}}^2
  =\norm{(P_v)_{00}-(P_v)_{00}^2}
  \leq r(1-r).
\end{align*}
Combining the last two bounds proves \eqref{eq:bound-||F||} and finishes
the proof.
\end{proof}

\section{Rapid mixing via local spectral contraction}\label{sec:general}
In this section, we will prove \Cref{thm:general}.
Let $G = (V,E)$ be a graph with maximum degree $\Delta$ and girth at least $5$.
Recall that $A \in [0,1]^{q\times q}$ is the symmetric interaction matrix and $\*\lambda \in \mathbb{R}_{\geq 0}^q$ is the external field.
The Gibbs distribution $\mu_G$ with interaction matrix $A$ and external field $\*\lambda$ on $G$ is defined by
\begin{align*}
  \mu_G(\sigma) \propto \prod_{uv \in E} A_{\sigma_u,\sigma_v} \prod_{v\in V} \lambda_{\sigma_v}, \quad \forall \sigma \in [q]^V.
\end{align*}

Fix $v \in V$ and a feasible pinning $\tau \in [q]^{V\setminus N[v]}$.
Then, for convenience, let
\begin{align} \label{eq:def-nu-pi}
  \nu &\coloneqq \mu^\tau_{G,N[v]},
  \quad
  \nu_v \coloneqq \mu^\tau_{G,v},
  \quad \text{and} \quad
  p_i \coloneqq \mu^\tau_{G-v, i}, \quad \forall i\in N(v).
\end{align}
We also identify $N(v)$ with $[d]$, where $d = \abs{N(v)}$.

It follows directly that \Cref{cond-local-spectral-contraction} implies the following condition on $\nu, (p_i)_{i=1}^d$.

\begin{condition}\label{cond:alpha-eps-Delta-contraction}
  There are $\alpha \geq 1$ and $\epsilon \in (0,1]$ such that, with
  \[B \coloneqq \*1\*1^\intercal - A,\]
  the following conditions hold for $\nu_v$ and $(p_i)_{i=1}^d$:
  \begin{itemize}
  \item ($\alpha$-marginal bound) for every $i\in[d]$ and
    $\!c\in[q]$,
    \begin{align*}
      (Bp_i)(\!c)\leq \frac{\alpha}{\Delta}, \quad (B\nu_v)(\!c)\leq \frac{\alpha}{\Delta};
    \end{align*}
  \item ($\epsilon$-contraction) we have
    \begin{align*}
      \norm{
      \-{diag}\tp{\sum_{i=1}^d p_i}^{1/2}
      {B\,
      \-{diag}}\tp{\nu_v}^{1/2}
      }_{{2}}^2 \leq \frac{1-\epsilon}{\Delta}.
    \end{align*}
  \end{itemize}
\end{condition}

The main result of this section is the following lemma.

\begin{lemma}\label{lem:block-estimate-via-contraction}
  Suppose \Cref{cond:alpha-eps-Delta-contraction} holds, {$d\geq1$,} and $\Delta\geq2\alpha$.
  Then, on the inner-product space $(\*1^\perp, \inner{\cdot}{\cdot}_{\nu})$, we have
  \begin{align}
    \label{eq:general-P00}
    \norm{(P_v)_{00}} &\leq O_{\alpha}\tp{1/\Delta}, \\
    \label{eq:general-P10}
    \norm{(P_v)_{10}}^2 &\leq \tp{1 + O_\alpha(1/\Delta)} \cdot \frac{1-\epsilon}{\Delta} + O_\alpha\tp{\frac{\sqrt{\log(2d)}}{\Delta^{3/2}}}.
  \end{align}
\end{lemma}

We are now ready to prove \Cref{thm:general}.
\begin{proof}[Proof of \Cref{thm:general}]
For every $v\in V$, set
\begin{align*}
  \eta_v\coloneqq\frac{1-\epsilon/2}{\Delta}.
\end{align*}
We verify the two hypotheses of \Cref{thm:local-to-global} for these
parameters.

Fix $v\in V$ and a feasible pinning $\tau$ outside $N[v]$, and let
$d=\abs{N(v)}$.  If $d=0$, then the weighted dynamics on $N[v]=\set{v}$
has Laplacian $\+L_v$ and satisfies the Poincar{\'e} inequality with
constant one, which is at least $(1-\eta_v)/2$.  We may therefore assume
that $d\geq1$.

Use the notation in \eqref{eq:def-nu-pi}.  Since $G$ has girth at least
$5$, the induced graph on $N[v]$ is a star, and, conditional on the
center, its leaves are independent.  Moreover,
\Cref{cond-local-spectral-contraction} implies
\Cref{cond:alpha-eps-Delta-contraction}.  Hence
\Cref{lem:block-estimate-via-contraction} gives
\begin{align*}
  r\coloneqq\norm{(P_v)_{00}}
  &=O_\alpha\tp{\frac1\Delta},
  \\
  \norm{(P_v)_{10}}^2
  &\leq
  \tp{1+O_\alpha\tp{\frac1\Delta}}
  \frac{1-\epsilon}{\Delta}
  +O_\alpha\tp{
    \frac{\sqrt{\log(2d)}}{\Delta^{3/2}}
  }.
\end{align*}
In particular, by taking the lower bound on $\Delta$ sufficiently large
in terms of $\alpha$ and $\epsilon$, we have $r<1/5$ and
\begin{align*}
  \widehat\eta
  &\coloneqq\frac{\norm{(P_v)_{10}}^2}{(1-r)(1-4r)}
  \\
  &\leq
  \frac{1-\epsilon}{\Delta}
  +O_\alpha\tp{
    \frac1{\Delta^2}
    +\frac{\sqrt{\log(2\Delta)}}{\Delta^{3/2}}
  }
  \leq
  \frac{1-\epsilon/2}{\Delta}
  =\eta_v<1,
\end{align*}
where we used $d\leq\Delta$ and $r=O_\alpha(1/\Delta)$ in the first
inequality.

Applying \Cref{poincare-via-proj-norms} with $\widehat\eta$ shows that,
for every $f\perp\*1$ in $L^2(\nu)$,
\begin{align*}
  \frac{1-\widehat\eta}{2}\norm{f}^2
  \leq
  \inner{f}{\tp{
    \+L_v+\frac{1+\widehat\eta}{2}\+L_{N(v)}
  }f}.
\end{align*}
Since $\widehat\eta\leq\eta_v$ and $\+L_{N(v)}\succeq0$, this implies
\begin{align*}
  \frac{1-\eta_v}{2}\norm{f}^2
  \leq
  \inner{f}{\tp{
    \+L_v+\frac{1+\eta_v}{2}\+L_{N(v)}
  }f}.
\end{align*}
Thus the required weighted spectral-gap estimate holds uniformly over
all feasible pinnings $\tau$.

For every $v\in V$, the maximum-degree assumption gives
\begin{align*}
  \sum_{u\in N(v)}\eta_u
  \leq
  \Delta\cdot\frac{1-\epsilon/2}{\Delta}
  =1-\frac\epsilon2.
\end{align*}
The factorization defining $\mu_G$ also implies the conditional
independence property in \Cref{def:cond-indep}.  {Global irreducibility holds by hypothesis.} Therefore
\Cref{thm:local-to-global}, applied with global reserve $\epsilon/2$,
shows that the continuous-time Glauber Laplacian $\+L$ has spectral gap
at least $\epsilon/4$.

Finally, the Laplacian of the discrete-time Glauber dynamics is
$\+L/n$.  Its spectral gap is consequently at least
$\epsilon/(4n)$, which proves the theorem.
\end{proof}

In the remainder of this section, we prove \Cref{lem:block-estimate-via-contraction}.
In particular, \eqref{eq:general-P00} and \eqref{eq:general-P10} will be proved in \Cref{sec:general-P00} and \Cref{sec:general-P10}, respectively.

We will work in the inner product space
\begin{align*}
  (\*1^\perp, \inner{\cdot}{\cdot}_\nu).
\end{align*}
We will first introduce some basic definitions and facts.
{We write $v\!c$ for the assignment $\sigma_v=\!c$, so $\nu(v\!c)=\nu_v(\!c)$ and $\nu^{v\!c}$ denotes conditioning on this assignment.}
By the definition of $\nu$, we have
%
\begin{align} \label{eq:def-local-mu}
  \nu(v\!c, \sigma) \propto \lambda_{\!c}\prod_{i}A_{\!c,\sigma_i}p_i(\sigma_i), \quad \forall \!c \in [q], \sigma \in [q]^d.
\end{align}
In particular, this implies the \emph{tree recursion}:
\begin{align*}
  \nu(v\!c) = g_{\!c}(\*p) \coloneqq \frac{\lambda_{\!c} \prod_i (Ap_i)(\!c)}{\sum_{\!b} \lambda_{\!b} \prod_i (Ap_i)(\!b)}.
\end{align*}
If the center $v$ is fixed to color $\!c$, then the posterior distribution on leaf $i$ is
\begin{align*}
  \nu^{v\!c}(i\!a) \coloneqq \frac{A_{\!c,\!a}p_i(\!a)}{(Ap_i)(\!c)}, \quad \forall \!a \in [q].
\end{align*}

\begin{definition}\label{def:L}
  For the product distribution $\rho = \otimes_i p_i$ on the leaves, define the following function on $\-{supp}(\rho)$:
  \begin{align*}
    {L_\rho(x) \coloneqq \frac{\nu_{N(v)}(x)}{\rho(x)}, \quad \forall x \in \-{supp}(\rho).}
  \end{align*}
  When the context is clear, we drop the subscript $\rho$ from $L_\rho$ and directly use $L$.

  Moreover, we define similar ratios for conditional distributions:
  \begin{align*}
    \forall \!c,i, \quad L^{\!c}_i(\!a) &\coloneqq {\frac{A_{\!c,\!a}}{(Ap_i)(\!c)}}, \quad \!a\in [q],\\
    {\forall \!c, \quad L^{\!c}(x)} &{\coloneqq \prod_i L_i^{\!c}(x_i) = \prod_{j=1}^d \frac{A_{\!c,x_j}}{(Ap_j)(\!c)}, \quad \forall x \in \mathrm{supp}(\rho).}
  \end{align*}
  {Under the standing assumptions, $(Ap_i)(\!c)\geq1-\alpha/\Delta\geq1/2$, and $L_i^{\!c}(\!a)=\nu^{v\!c}(i\!a)/p_i(\!a)$ whenever $p_i(\!a)>0$ and $\nu(v\!c)>0$.}
  For a fixed $i$, we measure the influence from $v$ to $i$ by $\Psi_i \in \mathbb{R}^{q\times q}$:
\begin{align*}
  \Psi_{i}(\!c,\!d) \coloneqq \mathrm{Cov}_{p_i}(L^{\!c}_i,L^{\!d}_i).
\end{align*}
\end{definition}


\begin{proposition} \label{prop:Psi-i}
  For each fixed $i$,
  \begin{align*}
    \Psi_i = \mathrm{diag}(\*1-Bp_i)^{-1}B\tp{\mathrm{diag}(p_i) - p_ip_i^\intercal}B \, \mathrm{diag}(\*1-Bp_i)^{-1}.
  \end{align*}
\end{proposition}
\begin{proof}
  Fix $\!c,\!d$. By definition, we have
  \begin{align*}
    \Psi_i(\!c,\!d)
    &= \E[p_i]{L^{\!c}_i L^{\!d}_i} - 1 
    =   \frac{\sum_{\!b}p_i(\!b)A_{\!c,\!b}A_{\!d,\!b}}{(Ap_i)(\!c) (Ap_i)(\!d)} - 1 \\
    \tag{$B = \*1\*1^\intercal - A$}
    &=   \frac{\sum_{\!b}p_i(\!b)B_{\!c,\!b}B_{\!d,\!b} -(Bp_i)(\!c)(Bp_i)(\!d)}{(1-(Bp_i)(\!c))(1 - (Bp_i)(\!d))}.
  \end{align*}
  This finishes the proof.
\end{proof}

\begin{corollary} \label{cor:Psi-related-bounds}
  Under \Cref{cond:alpha-eps-Delta-contraction}, if $\Delta\geq2\alpha$, then
  \begin{itemize}
  \item for every $i$, $\!c,\!d$, $\abs{\Psi_i(\!c,\!d)} = O_\alpha(\Delta^{-1})$;
  \item for every $i$, $\norm{\abs{\Psi_i} \nu_v}_\infty = O_{\alpha}(\Delta^{-2})$,
  \end{itemize}
  where $\norm{\cdot}_\infty$ is the {infinity} norm, and $\abs{\cdot}$ means taking entrywise absolute values.
\end{corollary}
\begin{proof}
  For the first bound, by the $\alpha$-marginal bound, \Cref{prop:Psi-i}, and $B \leq 1$ entry-wise, 
  \begin{align*}
    \abs{\Psi_i(\!c,\!d)}
    &\leq (1-\alpha/\Delta)^{-2}\tp{(Bp_i)(\!c) + \frac{\alpha^2}{\Delta^2}} \\
    &\leq (1-\alpha/\Delta)^{-2}\tp{\frac{\alpha}{\Delta} + \frac{\alpha^2}{\Delta^2}} = O_\alpha(\Delta^{-1}).
  \end{align*}
  Similarly, for the second bound, fix $\!c$. By the $\alpha$-marginal bound and \Cref{prop:Psi-i}, we have
  \begin{align*}
    (\abs{\Psi_i} \nu_v)(\!c)
    &= \sum_{\!d} \nu(v\!d)\abs{\Psi_i(\!c,\!d)} \\
    &\leq (1-\alpha/\Delta)^{-2}\sum_{\!d} \nu(v\!d)\tp{\sum_{\!b}p_i(\!b)B_{\!c,\!b}B_{\!d,\!b} + \frac{\alpha^2}{\Delta^2}} \\
    &=(1-\alpha/\Delta)^{-2} \sum_{\!b}p_i(\!b)B_{\!c,\!b} (B\nu_v)(\!b) + O_\alpha(\Delta^{-2}) \\
    &\leq (1-\alpha/\Delta)^{-2} \cdot \frac{\alpha}{\Delta} (Bp_i)(\!c) + O_\alpha(\Delta^{-2}) = O_\alpha(\Delta^{-2}). \qedhere
  \end{align*}
\end{proof}

\begin{fact}\label{fact:var-matrix-form}
  Let $\pi$ be a distribution on $[q]$.
  Let $f:[q] \to \mathbb{R}$ be a function.
  Then
  \begin{align*}
    \Var[\pi]{f} = f^\intercal \tp{\mathrm{diag}(\pi) - \pi\pi^\intercal} f,
  \end{align*}
  where we treat $\pi$ and $f$ as vectors in $\mathbb{R}^q$.
\end{fact}

\begin{proposition}\label{prop:normalizer-conditional-tail}
Suppose $1\leq d\leq\Delta$ and the $\alpha$-marginal bound in
\Cref{cond:alpha-eps-Delta-contraction} holds.  Let
$\rho=\bigotimes_{i=1}^d p_i$ and $L=L_\rho$.
If $\Delta\geq2\alpha$, then there
are constants $C_{\alpha},c_{\alpha}>0$ such that, for
every $t\geq2C_{\alpha}d/\Delta^2$,
\begin{align}
  \sup_{\!r:\,\nu(v\!r)>0}
  \Pr[\nu^{v\!r}]{\abs{L-1}>t}
  \leq
  2\exp\tp{-c_{\alpha}\frac{\Delta^2}{d}t^2}
  \leq
  2\exp\tp{-c_{\alpha}\Delta t^2}.
  \label{eq:normalizer-conditional-tail}
\end{align}
\end{proposition}

\begin{proof}[Proof of \Cref{prop:normalizer-conditional-tail}]
Center colors of zero $\nu$-mass are omitted below.
For every $i$ and $\!c$, by the $\alpha$-marginal bound and $\Delta \geq 2\alpha$,
\begin{align}
  1-\frac{\alpha}{\Delta}
  \leq
  (Ap_i)(\!c)
  =1-(Bp_i)(\!c)
  \leq 1,
  \qquad
  \frac{\alpha}{\Delta}\leq\frac12.
  \label{eq:normalizer-z-bounds}
\end{align}
Moreover, \eqref{eq:normalizer-z-bounds} and
$0\leq A_{\!c,\!a}\leq1$ give
\begin{align}
  0\leq L^{\!c},L
  \leq
  \tp{1-\frac{\alpha}{\Delta}}^{-d}
  \leq
  \exp\tp{\frac{2\alpha d}{\Delta}}
  \leq
  \exp(2\alpha).
  \label{eq:normalizer-uniform-bound}
\end{align}

Suppose that two leaf configurations differ only at coordinate $i$,
where the leaf color changes from $\!a$ to $\!b$.  Then
\begin{align*}
  \abs{L^{\!c}(x)-L^{\!c}(x')}
  &\leq
  \tp{1-\frac{\alpha}{\Delta}}^{-d}
  \abs{A_{\!c,\!a}-A_{\!c,\!b}}
  \\
  &\leq
  \tp{1-\frac{\alpha}{\Delta}}^{-d}
  \tp{B_{\!c,\!a}+B_{\!c,\!b}}.
\end{align*}
Thus the symmetry of $B$ and the center interaction bound imply
\begin{align}
  \abs{L(x)-L(x')}
  &\leq
  \tp{1-\frac{\alpha}{\Delta}}^{-d}
  \sum_{\!c}\nu(v\!c)
  \tp{B_{\!c,\!a}+B_{\!c,\!b}}
  \notag\\
  \notag
  &= \tp{1-\frac{\alpha}{\Delta}}^{-d}\tp{(B\nu_v)(\!a) + (B\nu_v)(\!b)}\\
  &\leq
  \frac{2\alpha\exp(2\alpha)}{\Delta}.
  \label{eq:normalizer-bounded-difference}
\end{align}
We claim that there is $C_\alpha > 0$ such that for every $\!d$,
\begin{align}
  \abs{\E[\nu^{v\!d}]{L}-1}
  \leq
  \frac{C_{\alpha}d}{\Delta^2}.
  \label{eq:normalizer-conditional-mean-target}
\end{align}
Since the leaves are independent under every $\nu^{v\!d}$, \Cref{prop:normalizer-conditional-tail} now follows from \eqref{eq:normalizer-bounded-difference},
\eqref{eq:normalizer-conditional-mean-target}, and McDiarmid's inequality;
the second inequality in \eqref{eq:normalizer-conditional-tail} uses
$d\leq\Delta$.

It only remains to prove \eqref{eq:normalizer-conditional-mean-target}.
Note that 
\begin{align*}
  \E[\nu^{v\!d}]{L}
  &= \sum_{x} \nu^{v\!d}(x) \sum_{\!c} \nu(v\!c) L^{\!c}(x) \\
  &= \sum_{\!c}\nu(v\!c) \sum_{x} \rho(x) L^{\!d}(x) L^{\!c}(x) \\
  \tag{independent}
  &= \sum_{\!c}\nu(v\!c) \prod_{i} (1 + \Psi_i(\!d,\!c)).
\end{align*}
This implies that
\begin{align*}
  \abs{\E[\nu^{v\!d}]{L}-1}
  \tag{triangle ineq.}
  &\leq \sum_{\!c}\nu(v\!c){\abs{\prod_{i} (1 + \Psi_i(\!d,\!c))-1}} \\
  &\leq \sum_{\!c}\nu(v\!c)\exp\tp{\sum_i \abs{\Psi_i(\!d,\!c)}} \sum_i \abs{\Psi_i(\!d,\!c)} \\
  \tag{\Cref{cor:Psi-related-bounds}, 1st bullet}
  &= \sum_i(\abs{\Psi_i}\nu_v)(\!d) \cdot O_\alpha(1) \\
  \tag{\Cref{cor:Psi-related-bounds}, 2nd bullet}
  &= \frac{O_\alpha(1)d}{\Delta^2}.
\end{align*}
Here the second inequality uses the following bound:
\begin{align*}
  \abs{\prod_i(1+x_i)-1}
  \leq
  \exp\tp{\sum_i\abs{x_i}}
  \sum_i\abs{x_i}.
\end{align*}
This finishes the proof.
\end{proof}


\subsection{Bound for $\norm{(P_v)_{00}}$} \label{sec:general-P00}
In this subsection, we prove \eqref{eq:general-P00}.
We will need the following result.
\begin{observation} \label{obs:PvhL-on-i}
  For $h \in K_0$, let $F \coloneqq P_v h \cdot L$ be a function on the leaves.
  Fix $i\in[d]$ and $x \in \mathrm{supp}(\rho)$. Then
  \begin{align*}
    F(x_{\sim i}, x_i) = (Aw_{i,x_{\sim i}})(x_i),
  \end{align*}
  where the vector $w_{i,x_{\sim i}} \in \mathbb{R}^q$ is defined as
  \begin{align*}
    w_{i,x_{\sim i}}(\!c)
    \coloneqq
    \frac{\nu_v(\!c)h(\!c)}{(Ap_i)(\!c)}
    \prod_{j\ne i}L^{\!c}_j{(x_j)},\quad \forall \!c.
  \end{align*}
\end{observation}
\begin{proof}
  We have
  \begin{align*}
    F(x)
    \coloneqq P_vh(x)L(x) &=\sum_{\!c}\nu_v(\!c)h(\!c)L^{\!c}(x) \\
    &= \sum_{\!c} A_{x_i,\!c} \cdot \frac{\nu(v\!c)h(\!c)}{(Ap_i)(\!c)} \prod_{j\neq i} L^{\!c}_j{(x_j)}. \qedhere
  \end{align*}
\end{proof}

\begin{proof}[Proof of \eqref{eq:general-P00}]
If $d=0$, then $P_vh=\E[\nu]{h}=0$ for every $h\in K_0$, and the
claim is immediate.  We henceforth assume that $d\geq1$.
Since $\norm{(P_v)_{00}}\leq1$, values of $\Delta$ bounded in terms of
$\alpha$ can be absorbed into the implicit constant.  We may therefore
assume throughout that $\Delta$ is sufficiently large in terms of
$\alpha$; in particular, \Cref{prop:normalizer-conditional-tail}
applies with $t=1/2$.

Because $P_v$ is an orthogonal projection,
\begin{align*}
  \inner{h}{(P_v)_{00}h}
  =\inner{h}{P_vh}
  =\norm{P_vh}^2.
\end{align*}
Split this last norm according to the events $L\geq1/2$ and $L<1/2$.
On the first event, we claim the following bound
\begin{align} \label{eq:general-P00-second-moment-target}
  \E[\rho]{(P_vh \cdot L)^2} \leq C_\alpha\frac{d}{\Delta^2} \norm{h}^2.
\end{align}
This implies that
\begin{align*} 
  \E[\nu]{
    \*1_{\set{L\geq1/2}}(P_vh)^2
  }
  =
  \E[\rho]{
    \*1_{\set{L\geq1/2}}\frac{(P_vh \cdot L)^2}{L}
  }
  \leq
  C_\alpha\frac{d}{\Delta^2}\norm{h}^2.
\end{align*}
On the complementary event, conditional Jensen gives
$(P_vh)^2\leq P_v(h^2)$.  Therefore,
\Cref{prop:normalizer-conditional-tail}, applied with $t=1/2$, gives
\begin{align*}
  \E[\nu]{
    \*1_{\set{L<1/2}}(P_vh)^2
  }
  &\leq
  \sum_{\!c}\nu(v\!c) h(\!c)^2
  \Pr[\nu^{v\!c}]{L<1/2}
  \\
  &\leq
  2\exp\tp{-c_\alpha\frac{\Delta^2}{d}}
  \norm{h}^2.
\end{align*}
Combining the two events and taking the supremum over nonzero
$h\in K_0$, we have
\begin{align*}
  \norm{(P_v)_{00}}
  \leq
  C_\alpha\frac{d}{\Delta^2}
  +2\exp\tp{-c_\alpha\frac{\Delta^2}{d}}
  =O_\alpha\tp{\frac1\Delta},
\end{align*}
where the last equality uses $d\leq\Delta$.  This proves
\eqref{eq:general-P00}.

It only remains to prove \eqref{eq:general-P00-second-moment-target}.
Fix $h\in K_0$.  Since $h$ depends only on the center, we write
$h(\!c)$ for its value at the center color $\!c$.  Set
\begin{align*}
  \rho&\coloneqq\bigotimes_{i=1}^d p_i{.}
\end{align*}
Center colors of zero $\nu$-mass are omitted below.  The
$\alpha$-marginal bound in
\Cref{cond:alpha-eps-Delta-contraction} gives
\begin{align}
  1-\frac{\alpha}{\Delta}
  \leq
  (Ap_i)(\!c)
  =1-(Bp_i)(\!c)
  \leq1.
  \label{eq:general-P00-z-bound}
\end{align}
For $x\in\-{supp}(\rho)$, recall
\begin{align*}
  F(x)
  &\coloneqq{}
  P_vh(x)L(x)
  =\sum_{\!c}\nu_v(\!c)h(\!c)L^{\!c}(x).
\end{align*}
Since $h$ is centered,
\begin{align*}
  \E[\rho]{F}
  =\sum_{\!c}\nu_v(\!c)h(\!c)
  =0.
\end{align*}
Tensorization of variance under $\rho$ therefore yields
\begin{align}
  \E[\rho]{F^2}
  \leq
  \sum_{i=1}^d
  \E[\rho]{
    \Var[\rho]{F\mid X_{\sim i}}
  }.
  \label{eq:general-P00-tensorization}
\end{align}

By \Cref{obs:PvhL-on-i}, we have
\begin{align*}
  F(x_{\sim i}, x_i) =(Aw_{i,x_{\sim i}})(x_i), \quad \text{where} \quad
  w_{i,x_{\sim i}}(\!c)
  \coloneqq
  \frac{\nu_v(\!c)h(\!c)}{(Ap_i)(\!c)}
  \prod_{j\ne i}L^{\!c}_j{(x_j)},\quad \forall \!c.
\end{align*}
Applying
\Cref{fact:var-matrix-form} with distribution $p_i$ and function
$Aw_{i,x_{\sim i}}$, we obtain
\begin{align}
  \nonumber
  &\Var[\rho]{F\mid X_{\sim i}=x_{\sim i}}
  \\
  \nonumber
  &\qquad=
  w_{i,x_{\sim i}}^\intercal
  A\tp{\-{diag}(p_i)-p_ip_i^\intercal}A
  w_{i,x_{\sim i}}
  \\
  \nonumber
  &\qquad=
  w_{i,x_{\sim i}}^\intercal
  B\tp{\-{diag}(p_i)-p_ip_i^\intercal}B
  w_{i,x_{\sim i}}
  \\
  &\qquad\leq
  w_{i,x_{\sim i}}^\intercal
  {B\,\-{diag}}(p_i)B
  w_{i,x_{\sim i}}.
  \label{eq:general-P00-tensorization-B}  
\end{align}
Indeed, the middle identity uses
$A=\*1\*1^\intercal-B$ and
$\tp{\-{diag}(p_i)-p_ip_i^\intercal}\*1=0$.

Note that \eqref{eq:general-P00-z-bound} and $0\leq A\leq1$ give
\begin{align}
  \abs{L^{\!c}(x)}
  \leq
  \tp{1-\frac{\alpha}{\Delta}}^{-d}
  \leq\exp(2\alpha).
  \label{eq:general-P00-likelihood-bound}
\end{align}
Let $y(\!c)\coloneqq\nu_v(\!c)\abs{h(\!c)}$.  By
\eqref{eq:general-P00-likelihood-bound},
\begin{align*}
  \abs{w_{i,x_{\sim i}}(\!c)}
  \leq\exp(2\alpha)y(\!c).
\end{align*}
{Set $s\coloneqq\sum_{i=1}^d p_i$.} Since ${B\,\-{diag}}(p_i)B$ has nonnegative entries,
\eqref{eq:general-P00-tensorization-B} and \eqref{eq:general-P00-tensorization} give
\begin{align}
  \E[\rho]{F^2}
  \nonumber
  &\leq
  \exp(4\alpha)
  y^\intercal
  {B\,\-{diag}}\tp{\sum_{i=1}^d p_i}B
  y \\
  \nonumber
  &=
  \exp(4\alpha)
  \abs{h}^\intercal\mathrm{diag}(\nu_v)
  {B\,\-{diag}}\tp{\sum_{i=1}^d p_i}{B\,
  \mathrm{diag}}(\nu_v)\abs{h} \\
  &\leq 
    \exp(4\alpha) \cdot
    \norm{
    \-{diag}(s)^{1/2}
    {B\,
    \-{diag}}(\nu_v)^{1/2}
    }_{2}^2 \norm{h}^2,
  \label{eq:general-P00-second-moment-preliminary}
\end{align}
where the $2$-norm is the {operator norm induced by the} Euclidean norm.

The $\alpha$-marginal bound implies
\begin{align*}
  (Bs)(\!c)
  \leq\frac{\alpha d}{\Delta},
  \qquad
  (B\nu_v)(\!c)
  \leq\frac{\alpha}{\Delta},
  \quad \forall \!c \in [q].
\end{align*}
Thus, for every $x\in\mathbb R^q$, weighted
Cauchy--Schwarz and the symmetry of $B$ give
\begin{align*}
  &\norm{
    \-{diag}(s)^{1/2}
    {B\,
    \-{diag}}(\nu_v)^{1/2}x
  }_2^2
  \\
  &\qquad=
  \sum_{\!a}s(\!a)
  \tp{
    \sum_{\!c}B_{\!a,\!c}
    \sqrt{\nu_v(\!c)}x(\!c)
  }^2
  \\
  &\qquad\leq
  \sum_{\!a}s(\!a)(B\nu_v)(\!a)
  \sum_{\!c}B_{\!a,\!c}x(\!c)^2
  \\
  &\qquad\leq
  \frac{\alpha}{\Delta}
  \sum_{\!c}(Bs)(\!c)x(\!c)^2
  \leq
  \frac{\alpha^2d}{\Delta^2}\norm{x}_2^2.
\end{align*}
Hence
\eqref{eq:general-P00-second-moment-preliminary} yields
\begin{align}
  \E[\rho]{F^2}
  \leq
  C_\alpha\frac{d}{\Delta^2}\norm{h}^2,
  \label{eq:general-P00-second-moment}
\end{align}
as claimed in \eqref{eq:general-P00-second-moment-target}.
\end{proof}

\subsection{Bound for $\norm{(P_v)_{10}}$} \label{sec:general-P10}
In this subsection, we prove \eqref{eq:general-P10}.
Let $\rho=\bigotimes_{i=1}^d p_i$, and let $L=L_\rho$.
For $h \in K_0$, we have
\begin{align*}
  \norm{(P_v)_{10}h}
  &= \norm{\-{proj}_{K_1} P_v h} \\
  \tag{triangle ineq.}
  &\leq \norm{\-{proj}_{K_1} (P_v h \cdot L)} + \norm{\-{proj}_{K_1} (P_v h \cdot (1-L))} \\
  &\leq \norm{\-{proj}_{K_1} (P_v h \cdot L)} + \norm{P_v h \cdot (1-L)},
\end{align*}
where in the last line, we use the fact that $\norm{\-{proj}_{K_1}} \leq 1$.

\begin{proposition} \label{prop:master-form}
{Let $h\in K_0$.} For every $i, \!c, \!d$, we define
\begin{align*}
  R_i(\!c,\!d) \coloneqq \prod_{j\neq i} \tp{1 + \Psi_j(\!c,\!d)}.
\end{align*}
For every $i,\!c$, define the vector $u_{i,\!c}\in\mathbb R^q$ by setting
$u_{i,\!c}(\!d)=0$ when $\nu(v\!d)=0$; otherwise, set
\begin{align*}
  u_{i,\!c}(\!d) \coloneqq \frac{h(\!d) \nu(v\!d) R_i(\!c,\!d)}{{(Ap_i)(\!d)}}{.}
\end{align*}
Then
  \begin{align*}
    \norm{\operatorname{proj}_{K_1}(P_v h \cdot L)}^2
    &= \sum_{i, \!c} \nu(v\!c) u_{i,\!c}^\intercal A \tp{ \-{diag}(\nu^{v\!c}_i)  - \nu^{v\!c}_i(\nu^{v\!c}_i)^{\intercal}} A u_{i,\!c}.
  \end{align*}
\end{proposition}
The proof of \Cref{prop:master-form} is a tedious but straightforward calculation from the definitions and is deferred to \Cref{sec:master-form}.

\begin{proposition}\label{cor:related-to-contraction}
Suppose \Cref{cond:alpha-eps-Delta-contraction} holds.
Then, for $\Delta\geq2\alpha$,
every $h\in K_0$ satisfies
\begin{align}
  \norm{\operatorname{proj}_{K_1}(P_vhL)}^2
  &\leq
  \tp{1+O_{\alpha}\tp{\frac1\Delta}}
  \cdot \frac{1-\epsilon}{\Delta}\cdot 
  \norm{h}^2.
  \label{eq:row-bounded-B-bound}
\end{align}
\end{proposition}

\begin{proposition}\label{prop:row-bounded-normalizer}
  Suppose $1\leq d\leq\Delta$,
  \Cref{cond:alpha-eps-Delta-contraction} holds, and
  $\Delta\geq2\alpha$.  Then every $h\in K_0$ satisfies
\begin{align}
  \norm{P_vh(1-L)}^2
  \leq
  O_{\alpha}\tp{
    \frac{\log(2d)}{\Delta^2}
  }
  \norm{h}^2.
  \label{eq:row-bounded-normalizer-bound}
\end{align}
\end{proposition}
The proofs of \Cref{cor:related-to-contraction} and \Cref{prop:row-bounded-normalizer} are deferred to
\Cref{sec:related-to-contraction} and \Cref{sec:row-bounded-normalizer}, respectively.
Now, we are ready to prove \eqref{eq:general-P10}.

\begin{proof}[Proof of \eqref{eq:general-P10}]
  Combine \Cref{cor:related-to-contraction} and \Cref{prop:row-bounded-normalizer}.
\end{proof}

\subsection{Proof of \Cref{cor:related-to-contraction}}
\label{sec:related-to-contraction}

\begin{proof}
  Fix $h\in K_0$. Since $A=\*1\*1^\intercal-B$ and
  \[\tp{\mathrm{diag}(\nu^{v\!c}_i) - \nu^{v\!c}_i(\nu^{v\!c}_i)^\intercal}\*1=\*0,\]
 \Cref{prop:master-form} gives
\begin{align*}
  \norm{\operatorname{proj}_{K_1}(P_vhL)}^2
  &=
  \sum_{i,\!c}\nu(v\!c)
  u_{i,\!c}^\intercal B \tp{\mathrm{diag}(\nu^{v\!c}_i) - \nu^{v\!c}_i(\nu^{v\!c}_i)^\intercal} B u_{i,\!c} \\
  &\leq 
  \sum_{i,\!c}\nu(v\!c)
  \abs{u_{i,\!c}^\intercal} B\, \mathrm{diag}(\nu^{v\!c}_i) B \abs{u_{i,\!c}} \\
  &\leq 
  \tp{1 + O_\alpha(\Delta^{-1})} \sum_{i,\!c}\nu(v\!c)
  \abs{u_{i,\!c}^\intercal} B\, \mathrm{diag}(p_i) B \abs{u_{i,\!c}},
\end{align*}
where in both inequalities, we use the fact that $B$ is nonnegative and
\begin{align*}
  \nu^{v\!c}_i(\!a) &= \frac{A_{\!c,\!a}p_i(\!a)}{(Ap_i)(\!c)} 
  = \frac{A_{\!c,\!a} p_i(\!a)}{1-(Bp_i)(\!c)}
  \leq \frac{p_i(\!a)}{1-\alpha/\Delta}.
\end{align*}
The last inequality follows from $A \leq 1$ and $(Bp_i)(\!c) \leq \alpha/\Delta$.
Plugging in the definition of $u_{i,\!c}$ from \Cref{prop:master-form}, we have
\begin{align*}
  &\norm{\operatorname{proj}_{K_1}(P_vhL)}^2 \\
  &\leq \tp{1 + O_\alpha(\Delta^{-1})} \sum_{i,\!c} \nu(v\!c) \sum_{\!d,\!a} \abs{u_{i,\!c}(\!d)} ({B\,\mathrm{diag}}(p_i)B)_{\!d,\!a} \abs{u_{i,\!c}(\!a)} \\
  &= \tp{1 + O_\alpha(\Delta^{-1})} \sum_{i,\!c} \nu(v\!c) \sum_{\!d,\!a} \frac{\abs{h(\!d)\nu(v\!d)}\abs{h(\!a)\nu(v\!a)}R_i(\!c,\!d)R_i(\!c,\!a)}{(Ap_i)(\!d)(Ap_i)(\!a)} ({B\,\mathrm{diag}}(p_i)B)_{\!d,\!a} \\
  &= \tp{1 + O_\alpha(\Delta^{-1})} \sum_{i} \sum_{\!d,\!a} \frac{\abs{h(\!d)\nu(v\!d)}\abs{h(\!a)\nu(v\!a)}({B\,\mathrm{diag}}(p_i)B)_{\!d,\!a}}{(Ap_i)(\!d)(Ap_i)(\!a)} \sum_{\!c}\nu(v\!c) R_i(\!c,\!d)R_i(\!c,\!a).
\end{align*}
We claim that for every $\!d$ and $i\in [d]$,
\begin{align} \label{eq:row-bounded-averaged-R-target}
  \sum_{\!c}\nu(v\!c)R_i(\!c,\!d)^2 \leq 1+O_\alpha(\Delta^{-1}).
\end{align}
By the Cauchy--Schwarz inequality, this implies that
\begin{align*}
  \sum_{\!c}\nu(v\!c) R_i(\!c,\!d)R_i(\!c,\!a)
  &\leq  \tp{
    \sum_{\!c}\nu(v\!c)R_i(\!c,\!d)^2
  }^{1/2}
  \tp{
    \sum_{\!c}\nu(v\!c)R_i(\!c,\!a)^2
  }^{1/2}
  \leq 1+O_\alpha(\Delta^{-1}).
\end{align*}
Together, we have
\begin{align*}
  \norm{\operatorname{proj}_{K_1}(P_vhL)}^2
  &\leq \tp{1 + O_\alpha(\Delta^{-1})} \sum_{i} \sum_{\!d,\!a} \frac{\abs{h(\!d)\nu(v\!d)}\abs{h(\!a)\nu(v\!a)}({B\,\mathrm{diag}}(p_i)B)_{\!d,\!a}}{(Ap_i)(\!d)(Ap_i)(\!a)} \\
  &= \tp{1 + O_\alpha(\Delta^{-1})} \sum_{i} \sum_{\!d,\!a} \frac{\abs{h(\!d)\nu(v\!d)}\abs{h(\!a)\nu(v\!a)}({B\,\mathrm{diag}}(p_i)B)_{\!d,\!a}}{(1-(Bp_i)(\!d))(1-(Bp_i)(\!a))} \\
  &\leq \tp{1 + O_\alpha(\Delta^{-1})} \sum_{i} \sum_{\!d,\!a} \abs{h(\!d)\nu(v\!d)}\abs{h(\!a)\nu(v\!a)}({B\,\mathrm{diag}}(p_i)B)_{\!d,\!a} \\
  &= \tp{1 + O_\alpha(\Delta^{-1})} \sum_{i} \abs{h}^\intercal \mathrm{diag}(\nu_v) {B\,\mathrm{diag}}(p_i){B\, \mathrm{diag}}(\nu_v) \abs{h} \\
  &= \tp{1 + O_\alpha(\Delta^{-1})} \abs{h}^\intercal \mathrm{diag}(\nu_v) {B\,\mathrm{diag}}\tp{\textstyle \sum_i p_i}{B\, \mathrm{diag}}(\nu_v) \abs{h} \\
  &\leq \tp{1 + O_\alpha(\Delta^{-1})} \norm{ \-{diag}\tp{\sum_{i=1}^d p_i}^{1/2} {B\, \-{diag}}\tp{\nu_v}^{1/2}}_2^2 \norm{h}^2 \\
  \tag{$\epsilon$-contraction}
  &\leq \tp{1 + O_\alpha(\Delta^{-1})} \frac{1-\epsilon}{\Delta}\cdot \norm{h}^2.
\end{align*}
This finishes the proof of \Cref{cor:related-to-contraction}.
It only remains to prove \eqref{eq:row-bounded-averaged-R-target}.
By definition, 
\begin{align*}
  R_i(\!c,\!d)^2
  &\leq \prod_j\tp{1 + \abs{\Psi_j}(\!c,\!d)}^2 
  \leq \exp\tp{\sum_j 2\abs{\Psi_j}(\!c,\!d)} \\
  \tag{$\e^x \leq 1 + x\e^x$}
  &\leq 1 + \exp\tp{\sum_j 2\abs{\Psi_j}(\!c,\!d)} \cdot \tp{\sum_j 2\abs{\Psi_j}(\!c,\!d)}.
\end{align*}
Hence
\begin{align*}
  \sum_{\!c}\nu(v\!c)R_i(\!c,\!d)^2
  &\leq 1 + O_\alpha(1)\sum_j\sum_{\!c}\nu(v\!c)\abs{\Psi_j}(\!c,\!d) \\
  &=1+O_\alpha(1)\sum_j(\abs{\Psi_j}\nu_v)(\!d) \\
  \tag{\Cref{cor:Psi-related-bounds}}
  &= 1 + O_\alpha(d/\Delta^2)
  = 1 + O_\alpha(\Delta^{-1}).
\end{align*}
This finishes the proof.
\end{proof}

\subsection{Proof of \Cref{prop:row-bounded-normalizer}}
\label{sec:row-bounded-normalizer}

\begin{proof}
Fix $h\in K_0$.  Since $h$ depends only on the center, we write its
value at color $\!c$ as $h(\!c)$.  Center colors of zero $\nu$-mass are
omitted below.  Set
\begin{align*}
  \rho\coloneqq\bigotimes_{i=1}^d p_i.
\end{align*}
By \eqref{eq:def-local-mu}, \eqref{def-Pi}, and \Cref{def:L}, the
conditional leaf law $\nu^{v\!c}$ is absolutely continuous with
respect to $\rho$.  Hence, for every $x\in\-{supp}(\rho)$,
\begin{align}
  L(x)
  &=
  \sum_{\!c}\nu(v\!c)\frac{\nu^{v\!c}(x)}{\rho(x)},
  &
  P_vh(x)L(x)
  &=
  \sum_{\!c}\nu(v\!c)h(\!c)
  \frac{\nu^{v\!c}(x)}{\rho(x)}.
  \label{eq:normalizer-likelihood-identities}
\end{align}
Changing measure from the leaf marginal of $\nu$ to $\rho$ gives the
exact identity
\begin{align}
  \norm{P_vh(1-L)}^2
  =
  \E[\rho]{
    (P_vhL)^2\frac{(1-L)^2}{L}
  }.
  \label{eq:normalizer-exact-error}
\end{align}
The integrand on the right is defined to be zero when $L=0$; this is
consistent because \eqref{eq:normalizer-likelihood-identities} then
gives $P_vhL=0$.

It remains to estimate the right-hand side of
\eqref{eq:normalizer-exact-error}.  Values of $\Delta$ bounded in terms
of $\alpha$ can be absorbed into the implicit constant in
\eqref{eq:row-bounded-normalizer-bound}: indeed,
\eqref{eq:normalizer-uniform-bound} and the contraction of $P_v$ give
\begin{align*}
  \norm{P_vh(1-L)}
  \leq
  \tp{1+\exp(2\alpha)}\norm{h}.
\end{align*}
We may therefore assume that $\Delta$ is sufficiently large in terms
of $\alpha$.  Choose
\begin{align*}
  t
  \coloneqq
  C_\alpha
  \sqrt{
    \frac1\Delta
    +\frac{d\log(2d)}{\Delta^2}
  },
\end{align*}
where the constant is sufficiently large.  Uniformly over
$1\leq d\leq\Delta$, we then have $t\leq1/2$, and $t$ satisfies the
hypothesis of \Cref{prop:normalizer-conditional-tail}.  Moreover,
\begin{align*}
  \frac{\Delta^2}{d}t^2
  =
  C_\alpha^2
  \tp{
    \frac{\Delta}{d}+\log(2d)
  }
  \geq
  C_\alpha^2\log(2\Delta),
\end{align*}
where the last inequality uses
$\Delta/d\geq\log(\Delta/d)$.  After increasing $C_\alpha$ if
necessary, \Cref{prop:normalizer-conditional-tail} gives
\begin{align}
  \sup_{\!c:\,\nu(v\!c)>0}
  \Pr[\nu^{v\!c}]{\abs{L-1}>t}
  \leq
  \frac{2}{\Delta^3}.
  \label{eq:normalizer-tail-at-t}
\end{align}

On the event $\abs{L-1}\leq t$, we have $L\geq1/2$.  Hence
\eqref{eq:general-P00-second-moment-target} and
\eqref{eq:normalizer-exact-error} give
\begin{align}
  \E[\rho]{
    \*1_{\set{\abs{L-1}\leq t}}
    (P_vhL)^2\frac{(1-L)^2}{L}
  }
  &\leq
  2t^2\E[\rho]{(P_vhL)^2}
  \notag\\
  &\leq
  C_\alpha
  \frac{\log(2d)}{\Delta^2}\norm{h}^2.
  \label{eq:normalizer-good-event}
\end{align}
For the complementary event, weighted Cauchy--Schwarz and
\eqref{eq:normalizer-likelihood-identities} give, for every
$x\in\-{supp}(\rho)$,
\begin{align*}
  (P_vh(x)L(x))^2
  &=
  \tp{
    \sum_{\!c}\nu(v\!c)h(\!c)
    \frac{\nu^{v\!c}(x)}{\rho(x)}
  }^2
  \\
  &\leq
  \tp{
    \sum_{\!c}\nu(v\!c)
    \frac{\nu^{v\!c}(x)}{\rho(x)}
  }
  \tp{
    \sum_{\!c}\nu(v\!c)h(\!c)^2
    \frac{\nu^{v\!c}(x)}{\rho(x)}
  }
  \\
  &=
  L(x)
  \sum_{\!c}\nu(v\!c)h(\!c)^2
  \frac{\nu^{v\!c}(x)}{\rho(x)}.
\end{align*}
Together with \eqref{eq:normalizer-uniform-bound} and
\eqref{eq:normalizer-tail-at-t}, this yields
\begin{align}
  &\E[\rho]{
    \*1_{\set{\abs{L-1}>t}}
    (P_vhL)^2\frac{(1-L)^2}{L}
  }
  \notag\\
  &\qquad\leq
  C_\alpha
  \sum_{x\in\-{supp}(\rho)}
  \rho(x)\*1_{\set{\abs{L(x)-1}>t}}
  \sum_{\!c}\nu(v\!c)h(\!c)^2
  \frac{\nu^{v\!c}(x)}{\rho(x)}
  \notag\\
  &\qquad=
  C_\alpha
  \sum_{\!c}\nu(v\!c)h(\!c)^2
  \sum_{x\in\-{supp}(\rho)}
  \nu^{v\!c}(x)\*1_{\set{\abs{L(x)-1}>t}}
  \notag\\
  &\qquad=
  C_\alpha
  \sum_{\!c}\nu(v\!c)h(\!c)^2
  \Pr[\nu^{v\!c}]{\abs{L-1}>t}
  \notag\\
  &\qquad\leq
  \frac{C_\alpha}{\Delta^3}\norm{h}^2
  \leq
  C_\alpha\frac{\log(2d)}{\Delta^2}\norm{h}^2.
  \label{eq:normalizer-bad-event}
\end{align}
Combining \eqref{eq:normalizer-good-event} and
\eqref{eq:normalizer-bad-event} in
\eqref{eq:normalizer-exact-error} proves
\eqref{eq:row-bounded-normalizer-bound}.
\end{proof}

\subsection{Proof of \Cref{prop:master-form}}
\label{sec:master-form}

\begin{proof}
Fix $h\in K_0$.  In this proposition, $L$ denotes $L_\rho$ for the
reference product distribution
\begin{align*}
  \rho\coloneqq\bigotimes_{j=1}^d p_j.
\end{align*}
By \eqref{eq:def-local-mu}, the leaf marginal of $\nu$ is absolutely
continuous with respect to $\rho$.  Thus every sum containing a quotient
by $p_j(\!e)$ may be restricted to $\!e\in\-{supp}(p_j)$; we use this
convention throughout the proof.  Center colors of zero $\nu$-mass do
not contribute to any of the sums below and are omitted.  The
conditional leaf laws on such center fibers may be chosen arbitrarily,
and for every $\!d$ with $\nu(v\!d)=0$, we take
$u_{i,\!c}(\!d)=0$.  For $\nu(v\!d)>0$, \eqref{eq:def-local-mu}
implies that ${(Ap_j)(\!d)}>0$ for every $j$, so all the ratios used below
are well-defined.

We first verify that $h$ depends only on the center.  Since
$h\in\Ker(\+L_{N(v)})$, \eqref{def-LB} and the fact that every
$\+L_i$ is an orthogonal projection give
\begin{align*}
  0
  =\inner{h}{\+L_{N(v)}h}
  =\sum_{i=1}^d\inner{h}{\+L_i h}
  =\sum_{i=1}^d\norm{\+L_i h}^2.
\end{align*}
Hence $\+L_i h=0$, or equivalently $P_i h=h$, for every $i$.
Conditional independence in \eqref{eq:def-local-mu} then yields
\begin{align*}
  h=\prod_{i=1}^dP_i h=\E[\nu]{h\mid X_v}.
\end{align*}
We may therefore write $h(\!d)$ for the value of $h$ when
$X_v=\!d$.

Set $F\coloneqq P_vhL$.  By \eqref{def-Pi}, both $P_vh$ and $L$ depend only on
the leaf configuration, so the same is true of $F$.  The function $F$
need not be centered, so we apply
\Cref{def:Hoeffding-expand,prop:Hoeffding-expand} to
$F-\E[\nu]{F}$.  Since $K_1\subseteq\*1^\perp$, subtracting this
constant does not change the projection onto $K_1$.  For each
$i\in[d]$, its degree-one component is
\begin{align*}
  \+L_i\prod_{j\ne i}P_jF
  &=
  \E[\nu]{F\mid X_v,X_i}-\E[\nu]{F\mid X_v}.
\end{align*}
Here the constant subtracted from $F$ disappears because
$P_j\*1=\*1$ and $\+L_i\*1=\*0$.
Indeed, conditional independence of the leaves shows that the product
of the projections with $j\ne i$ averages out precisely those leaves,
and applying $P_i$ then averages out the last leaf.  Orthogonality of
the degree-one components now gives
\begin{align}
  \norm{\operatorname{proj}_{K_1}F}^2
  &=
  \sum_{i=1}^d
  \E[\nu]{
    \tp{
      \E[\nu]{F\mid X_v,X_i}
      -\E[\nu]{F\mid X_v}
    }^2
  }
  \notag\\
  &=
  \sum_{i=1}^d\sum_{\!c}
  \nu(v\!c)
  \Var[\nu^{v\!c}_i]{
    \E[\nu]{F(X)\mid X_v=\!c,X_i}
  }.
  \label{eq:master-fiber-variance}
\end{align}

By Bayes' rule, we have
\begin{align*}
  F(x) = \sum_{\!d}\nu_v(\!d)h(\!d)L^{\!d}(x),
\end{align*}
which implies that
\begin{align*}
   \E[\nu]{F(X)\mid X_v=\!c,X_i=\!a}
   &= \sum_{\!d}\nu(v\!d)h(\!d) \E[x\sim \nu^{v\!c}]{L^{\!d}_i(\!a) \prod_{j\neq i} L^{\!d}_j(x_j)} \\
  \tag{independence}
   &= \sum_{\!d}\nu(v\!d)h(\!d) L^{\!d}_i(\!a) \prod_{j\neq i} \E[\nu^{v\!c}_j]{L^{\!d}_j} \\
   &= \sum_{\!d}\nu(v\!d)h(\!d) L^{\!d}_i(\!a) \prod_{j\neq i} \E[p_j]{L^{\!c}_j L^{\!d}_j} \\
   &= \sum_{\!d}\nu(v\!d)h(\!d) L^{\!d}_i(\!a) \prod_{j\neq i} \tp{1 + \Psi_j(\!c,\!d)} \\
   &= \sum_{\!d}\nu(v\!d)h(\!d) \frac{A_{\!a,\!d}}{(Ap_i)(\!d)} R_i(\!c,\!d) = (Au_{i,\!c})(\!a).
\end{align*}
We finish the proof by applying \Cref{fact:var-matrix-form} to $\Var[\nu^{v\!c}_i]{g}$, where
\begin{align*}
  g(\!a) &\coloneqq \E[\nu]{F(X)\mid X_v=\!c,X_i=\!a}. \qedhere
\end{align*}
\end{proof}


\section{Application to the antiferromagnetic Potts model}
\label{sec:Potts}
In this section, we prove \Cref{thm:Potts} {(and hence \Cref{thm:main} by taking $\beta=0$)}.
Let $G=(V,E)$ be an $n$-vertex simple graph with maximum degree
$\Delta$ and girth at least $5$.  Fix $\delta\in(0,1)$,
$\beta\in[0,1]$, and an integer $q\geq2$ satisfying
\begin{align*}
  q\geq(1+\delta)(1-\beta)\Delta.
\end{align*}
Recall that the antiferromagnetic $q$-state Potts distribution
$\mu_G$ on $[q]^V$ is defined by
\begin{align*}
  \mu_G(\sigma)
  \propto
  \beta^{\abs{\set{uv\in E:\sigma_u=\sigma_v}}},
  \qquad \sigma\in[q]^V.
\end{align*}
Thus every edge whose endpoints have the same color contributes a
factor $\beta$, and there is no external field.  All conditioned laws
and marginals below are formed from $\mu_G$.

\begin{proof}[Proof of \Cref{thm:Potts}]
We finish the proof by applying \Cref{thm:general}.
To prove the spectral-gap bound, it suffices to verify
\Cref{cond-local-spectral-contraction}.
Put $\vartheta\coloneqq1-\beta$.  The interaction matrix and its deficit matrix
are
\begin{align*}
  A=\*1\*1^\intercal-\vartheta\mathbb I,
  \qquad
  B\coloneqq\*1\*1^\intercal-A=\vartheta\mathbb I.
\end{align*}
We verify \Cref{cond-local-spectral-contraction} with
\begin{align*}
  \alpha\coloneqq\frac1\delta,
  \qquad
  \epsilon\coloneqq\frac{\delta}{1+\delta}.
\end{align*}

Fix $v\in V$ and a feasible pinning $\tau$ outside $N[v]$, and put
$d\coloneqq\deg_G(v)$.  For $i\in N(v)$ and $\!a\in[q]$, let
$k_{i,\!a}$ be the number of pinned neighbors of $i$ outside $N[v]$
whose color is $\!a$.  Since $G$ has girth at least $5$, $N[v]$
induces a star.  Thus, writing
\begin{align*}
  p_i\coloneqq\mu^\tau_{G-v,i},
\end{align*}
we have
\begin{align*}
  p_i(\!a)
  =
  \frac{\beta^{k_{i,\!a}}}
       {\sum_{\!b=1}^q\beta^{k_{i,\!b}}}.
\end{align*}
When $\beta=0$, we use the convention $0^0=1$; feasibility of $\tau$
ensures that the denominator is positive.  Bernoulli's inequality and
$\sum_{\!b}k_{i,\!b}\leq\Delta-1$ give
\begin{align*}
  \sum_{\!b=1}^q\beta^{k_{i,\!b}}
  &=
  \sum_{\!b=1}^q(1-\vartheta)^{k_{i,\!b}}
  \geq
  q-\vartheta\sum_{\!b=1}^q k_{i,\!b}
  \geq
  q-\vartheta(\Delta-1).
\end{align*}
Consequently, if $\vartheta>0$, then the assumption on $q$ yields
\begin{align}
  \vartheta p_i(\!a)
  \leq
  \frac{\vartheta}{q-\vartheta(\Delta-1)}
  \leq
  \frac1{\delta\Delta+1}
  \leq
  \frac1{\delta\Delta}.
  \label{eq:Potts-leaf-interaction-bound}
\end{align}
The same conclusion is immediate when $\vartheta=0$.

For each $\!a\in[q]$, set
\begin{align*}
  w_{\!a}
  &\coloneqq%
  \prod_{i\in N(v)}\tp{1-\vartheta p_i(\!a)},
  &
  Z&\coloneqq\sum_{\!b=1}^q w_{\!b}.
\end{align*}
The tree recursion on the conditioned star gives
\begin{align*}
  \nu_v(\!a)
  \coloneqq
  \mu^\tau_{G,v}(\!a)
  =
  \frac{w_{\!a}}Z.
\end{align*}
Since $\prod_i(1-u_i)\geq1-\sum_i u_i$ for $u_i\in[0,1]$,
\begin{align*}
  Z
  \geq
  \sum_{\!a=1}^q\tp{1-\vartheta\sum_{i\in N(v)}p_i(\!a)}
  =q-\vartheta d.
\end{align*}
It follows, again trivially when $\vartheta=0$, that
\begin{align}
  \vartheta\nu_v(\!a)
  \leq
  \frac{\vartheta}{q-\vartheta d}
  \leq
  \frac1{\delta\Delta}.
  \label{eq:Potts-center-interaction-bound}
\end{align}
Because $B=\vartheta\mathbb I$, \eqref{eq:Potts-leaf-interaction-bound}
and \eqref{eq:Potts-center-interaction-bound} prove the
$\alpha$-bounded marginal condition.

It remains to verify the contraction condition.  The case
$\vartheta=0$ is immediate, so assume $\vartheta>0$.  For each
$\!a\in[q]$, put
\begin{align*}
  T_{\!a}\coloneqq\vartheta\sum_{i\in N(v)}p_i(\!a).
\end{align*}
By \eqref{eq:Potts-leaf-interaction-bound} and
$\log(1-u)\geq-u/(1-u)$ for $u\in[0,1)$,
\begin{align*}
  \log w_{\!a}
  &\geq
  -\tp{1+\frac1{\delta\Delta}}T_{\!a}.
\end{align*}
Since $\sum_{\!a}T_{\!a}=\vartheta d$, the arithmetic--geometric
mean inequality gives
\begin{align}
  Z
  \geq
  q\tp{\prod_{\!a=1}^q w_{\!a}}^{1/q}
  \geq
  q\exp\tp{
    -\tp{1+\frac1{\delta\Delta}}
    \frac{\vartheta d}{q}
  }.
  \label{eq:Potts-normalizer-lower-bound}
\end{align}
On the other hand, $w_{\!a}\leq\exp(-T_{\!a})$.  Therefore
$T_{\!a}\exp(-T_{\!a})\leq \e^{-1}$ and
\eqref{eq:Potts-normalizer-lower-bound} imply
\begin{align*}
  &\Delta\vartheta^2\nu_v(\!a)
    \sum_{i\in N(v)}p_i(\!a)
  \\
  &\qquad=
  \Delta\vartheta T_{\!a}\frac{w_{\!a}}Z
  \\
  &\qquad\leq
  \frac{\vartheta\Delta}{q}
  \exp\tp{
    \tp{1+\frac1{\delta\Delta}}
    \frac{\vartheta d}{q}-1
  }.
\end{align*}
For sufficiently large $\Delta$ in terms of $\delta$, we may assume
$\Delta\geq\delta^{-2}$.  Using $d\leq\Delta$ and
$\vartheta\Delta/q\leq1/(1+\delta)$, we then have
\begin{align*}
  \tp{1+\frac1{\delta\Delta}}
  \frac{\vartheta d}{q}
  \leq1.
\end{align*}
Hence, uniformly in $\!a$,
\begin{align*}
  \vartheta^2\nu_v(\!a)
  \sum_{i\in N(v)}p_i(\!a)
  \leq
  \frac1{(1+\delta)\Delta}
  =
  \frac{1-\epsilon}{\Delta}.
\end{align*}
Since $B=\vartheta\mathbb I$, the matrix in the contraction condition
is diagonal and therefore
\begin{align*}
  &\norm{
    \-{diag}\tp{\sum_{i\in N(v)}p_i}^{1/2}
    {B\,
    \-{diag}}(\nu_v)^{1/2}
  }_2^2
  \\
  &\qquad=
  \vartheta^2\max_{\!a\in[q]}\set{
    \nu_v(\!a)\sum_{i\in N(v)}p_i(\!a)
  }
  \leq
  \frac{1-\epsilon}{\Delta}.
\end{align*}
The estimates above are uniform over $v$ and all feasible pinnings
$\tau$.  Thus the Potts measure has local
$(\delta^{-1},\delta/(1+\delta))$-spectral contraction.

The Glauber dynamics is irreducible.  Indeed, when $\beta>0$ the Gibbs
measure has full support.  When $\beta=0$, the assumption on $q$ gives
$q\geq\Delta+2$ for sufficiently large $\Delta$ in terms of $\delta$,
and the proper-coloring Glauber dynamics is irreducible in this range {\cite[Section~2]{Jerrum1995}}.
After increasing the lower bound on $\Delta$ once more if necessary,
all the hypotheses of \Cref{thm:general} hold.  We conclude that the
discrete-time Glauber dynamics has spectral gap
{at least $\delta/(4(1+\delta)n)\geq\delta/(8n)$}.

Finally, consider the smallest positive mass of a configuration.  If
$\beta=0$, the measure is uniform on a set of at most $q^n$ proper
colorings, so
\begin{align*}
  \min_{\sigma\in\-{supp}(\mu)}\mu(\sigma)\geq q^{-n}.
\end{align*}
If $0<\beta\leq1$, then every configuration has unnormalized weight at
least $\beta^{\abs{E}}\geq\beta^{n\Delta/2}$, while the partition
function is at most $q^n$.  Hence
\begin{align*}
  \min_{\sigma\in[q]^V}\mu(\sigma)
  \geq
  \beta^{n\Delta/2}q^{-n}.
\end{align*}
Substituting these bounds and the spectral-gap estimate into the
standard worst-start mixing bound
\eqref{eq:spgap-implies-mixing-time} gives, {for $0<\varepsilon_{\mathrm{TV}}\leq1/2$,} respectively,
\begin{align*}
  t_{\mathrm{mix}}({\varepsilon_{\mathrm{TV}}})
  &=
  {O\tp{
    \frac{n^2\log q+n\log\frac1{\varepsilon_{\mathrm{TV}}}}{\delta}
  }},
  &&\beta=0,
  \\
  t_{\mathrm{mix}}({\varepsilon_{\mathrm{TV}}})
  &=
  {O\tp{
    \frac{n^2\log q+n^2\Delta\log\frac1\beta
    +n\log\frac1{\varepsilon_{\mathrm{TV}}}}{\delta}
  }},
  &&0<\beta\leq1.
\end{align*}
This proves the theorem.
\end{proof}

\bibliographystyle{alpha}
\begingroup
\footnotesize
\raggedright
\IfFileExists{refs.bib}{\bibliography{refs}}{\bibliography{../refs}}
\endgroup

\end{document}